\RequirePackage{fix-cm}
\documentclass[twocolumn]{svjour3}
\smartqed
\usepackage{graphicx,textcomp,xcolor}
\usepackage[linesnumbered,ruled,vlined]{algorithm2e}
\usepackage{amsmath,amssymb,amsfonts,mathptmx}
\usepackage{xspace,subfigure,url,enumitem,multirow,booktabs}
\spnewtheorem{observation}{Observation}{\bfseries}{\rmfamily}

\newcommand{\kw}[1]{{\ensuremath{\mathsf{#1}}}\xspace}
\newcommand{\stitle}[1]{\par\noindent\textbf{#1}}
\journalname{The VLDB Journal}
\begin{document}
\title{Density Decomposition on Hypergraphs}
\titlerunning{Density Decomposition on Hypergraphs}
\author{Xiaoyu Leng \and Hongchao Qin\textsuperscript{*}\thanks{\textsuperscript{*}Corresponding author.} \and Rong-Hua Li \and Guoren Wang}
\authorrunning{Xiaoyu Leng et al.}
\institute{Xiaoyu Leng \at Beijing Institute of Technology, Beijing, China\\
\email{lengxiaoyu@bit.edu.cn}
\and Hongchao Qin \at Beijing Institute of Technology, Beijing, China\\
\email{hcqin@bit.edu.cn}
\and Rong-Hua Li \at Beijing Institute of Technology, Beijing, China\\
\email{rhli@bit.edu.cn}
\and Guoren Wang \at Beijing Institute of Technology, Beijing, China\\
\email{wanggrbit@gmail.com}}
\date{}
\maketitle
\begin{abstract}
Hypergraphs capture higher-order interactions involving multiple entities. Existing decomposition models, such as $k$-core and nbr-$k$-core, rely on degree or neighborhood constraints and may overlook fine-grained density variations. We propose a $(k,\delta)$-dense subhypergraph model based on capped hyperedge contributions. For each parameter pair, the unique dense set is the inclusionwise smallest maximizer of a penalized capped hyperedge contribution. These dense sets form a nested hierarchy characterized by egalitarian $\delta$-orientations and hypergraph reachability. We develop three algorithms for dense subhypergraph search and two for complete density decomposition. Both decomposition algorithms use integral flow balancing: \kw{Flow\text{-}Decomp} extracts density levels through strongly connected component condensation, while \kw{Reverse\text{-}Decomp} uses direct reverse propagation. Once an egalitarian orientation is obtained, both recover all density levels in linear time. We further derive sandwich bounds that restrict computation to equivalent residual subproblems and develop incremental algorithms for maintaining the hierarchy under hyperedge insertions and deletions. Experiments on nine real-world hypergraphs, including datasets with over 15 million vertices, compare our methods with four established decomposition models. The results show that \kw{Reverse\text{-}Decomp} improves decomposition efficiency, incremental maintenance outperforms recomputation in the tested settings, and the proposed model captures density structures overlooked by degree-based decompositions.

\keywords{Hypergraphs \and Density decomposition \and Dense subhypergraph search \and Dynamic maintenance \and Network flow}
\end{abstract}
\section{Introduction}

Hypergraphs provide a natural abstraction for modeling higher-order relationships in which one interaction may involve more than two entities. Such structures arise in many data management and mining scenarios, including multi-party transactions, group interactions in social platforms, collaborative activities in recommendation systems, and higher-order relations in scientific domains. Compared with ordinary graphs, hypergraphs preserve the cardinality and semantics of group relations, enabling more faithful analysis of cohesive structures formed by multi-way interactions~\cite{DBLP:journals/corr/abs-2204-05646,DBLP:journals/corr/abs-2301-04235}. 
A fundamental task in hypergraph analysis is dense subhypergraph decomposition, which organizes vertices into nested cohesive sets for community detection, pattern discovery, anomaly detection, and structural summarization~\cite{DBLP:journals/pvldb/QinZLLYW25,9458645,DBLP:journals/pvldb/ArafatKRG23}.

Existing hypergraph decomposition models are mainly based on local degree-style constraints. The representative $k$-core model requires every vertex to participate in at least $k$ hyperedges~\cite{9458645}, while nbr-$k$-core further considers the number of neighboring vertices induced by incident hyperedges~\cite{DBLP:journals/pvldb/ArafatKRG23,DBLP:journals/pacmmod/ZhangYWLZL25}. These models are simple and scalable, but their criteria are still discrete proxies for density rather than density itself. In a hypergraph, a large hyperedge may connect many vertices while contributing weakly to local cohesion. Counting incident hyperedges or neighboring vertices cannot reliably distinguish genuinely dense sets from broad but loosely connected groups. Consequently, degree-based decompositions may produce coarse, discontinuous, or redundant layers, and may merge structurally different sets into the same level.
Figure~\ref{exam} illustrates this limitation. The $k$-core decomposition in Figure~\ref{exam}(b) divides the hypergraph into only two layers, where vertices in $C_1$ and $C_2$ satisfy degree requirements $1$ and $2$, respectively. However, this decomposition does not reveal the finer density variation inside the hypergraph. In contrast, the density-oriented decomposition in Figure~\ref{exam}(a) separates the same hypergraph into multiple density levels: $D_{4,4}$ corresponds to density level $4$, $D_{3,4}\setminus D_{4,4}$ corresponds to density level $3$, and $D_{1,4}\setminus D_{2,4}$ corresponds to density level $1$. This example shows that meaningful density levels may be invisible to degree-based decomposition.

\begin{figure}[t]
	\centering
	\subfigure[($k,4$)-density decomposition]{
		\label{4or}
		\includegraphics[width=0.46\columnwidth]{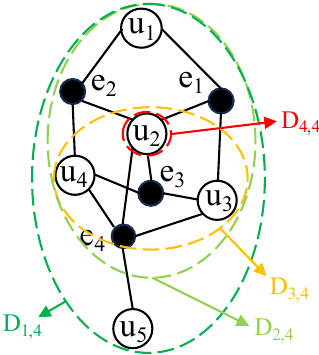}
	}
	\subfigure[Core decomposition]{
		\label{core}
		\includegraphics[width=0.46\columnwidth]{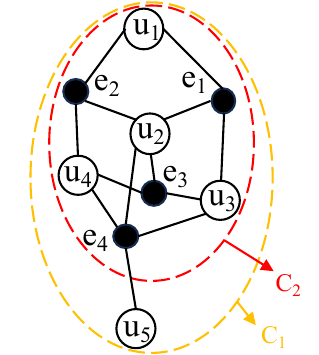}
	}
	
	\caption{Density decomposition and core decomposition.}
	\label{exam}
\end{figure}

Density levels are particularly useful when they reflect meaningful group structures. In legislative co-sponsorship hypergraphs, dense sets may reveal stable coalitions; in financial transaction networks, they may indicate suspicious coordination; and in biological interaction systems, cohesive modules often arise from overlapping multi-way relations~\cite{DBLP:journals/corr/abs-2204-05646}. These applications call for a scalable decomposition that directly captures density and distinguishes cohesive structures across levels.
However, hyperedges vary in size and may contribute to multiple vertices. Assigning their full size can overemphasize large hyperedges, while bounding their contributions requires a fair allocation among incident vertices. The challenge is to define a suitable contribution rule and efficiently construct the resulting density hierarchy.

In this paper, we propose the $(k,\delta)$-dense subhypergraph model based on the capped hyperedge contribution function $F_\delta$. For a vertex set $X$, each hyperedge contributes the portion of its budget $r_e=\min\{\delta,|e|\}$ that must be assigned within $X$. We define the unique dense set $D_{k,\delta}$ as the inclusionwise smallest maximizer of $F_\delta(X)-(k-1)|X|$. Supermodularity guarantees this uniqueness and nesting across parameter values, while egalitarian orientations and reachability support its computation.
We develop three search algorithms, \kw{DSM\text{-}PATH}, \kw{DSM\text{-}FLOW}, and \kw{DSM\text{-}FLOW+}, and two decomposition algorithms. \kw{Flow\text{-}Decomp} computes maximum reachable loads through SCC condensation, while \kw{Reverse\text{-}Decomp} uses direct reverse propagation. We further derive sandwich bounds that restrict density-level searches to equivalent residual subproblems and develop incremental algorithms for hyperedge insertions and deletions. Experiments on nine real-world hypergraphs demonstrate the efficiency and scalability of our algorithms and the advantages of incremental maintenance.

To summarize, our main contributions are as follows.
\begin{itemize}[leftmargin=10pt, topsep=0pt]
\item We propose the $(k,\delta)$-dense subhypergraph model, which enables density-oriented hierarchical decomposition of hypergraphs. In this model, $k$ represents the integer density level of each layer, while $\delta$ bounds the contribution of each hyperedge to the density. Unlike core-based and nbr-$k$-core decompositions~\cite{DBLP:journals/pvldb/ArafatKRG23,DBLP:journals/pacmmod/ZhangYWLZL25}, our model directly characterizes density variations induced by multi-way relationships.

\item We establish an orientation-independent supermodular characterization of $(k,\delta)$-dense subhypergraphs, proving the uniqueness and nesting of their density levels. Using egalitarian orientations, we develop three search algorithms: \kw{DSM\text{-}PATH} uses reversible hyperpaths, \kw{DSM\text{-}FLOW} employs integral flow to achieve $k$-stability, and \kw{DSM\text{-}FLOW+} improves initialization through greedy head assignment. For complete decomposition, we design \kw{Flow\text{-}Decomp}, which combines flow balancing with SCC condensation, and \kw{Reverse\text{-}Decomp}, which replaces SCC computation with direct reverse propagation. Both recover all density levels in $O(n+m+L)$ additional time after global balancing. We further establish a sandwich theorem that bounds each dense set before balancing and reduces its computation to an equivalent residual subproblem.

\item We evaluate the algorithms on nine real-world hypergraphs, including datasets with over 15 million vertices, considering search efficiency, decomposition time and memory, hierarchy granularity, scalability, and dynamic maintenance. We compare our methods with four established decomposition models on five datasets. \kw{Reverse\text{-}Decomp} consistently outperforms \kw{Flow\text{-}Decomp} on these datasets, achieving speedups of 3.11$\times$ and 2.39$\times$ on \kw{SB} and \kw{HB}, respectively.
\item The code is available at \url{https://github.com/xiaoyu-ll/IDH}.
\end{itemize}

\section{Problem Definition}

Let $\mathcal{H}=(V,E)$ be an undirected and unweighted hypergraph, where $V$ is the vertex set and $E$ is an indexed family of nonempty hyperedges. Each hyperedge $e\in E$ is a \emph{set} $e\subseteq V$, $|e|\ge1$. Identical hyperedges may occur as distinct indexed records, but a vertex cannot appear more than once within a single hyperedge. Let $n=|V|$ and $m=|E|$ denote the numbers of vertices and hyperedges, respectively. We use $d^E_{\max}$ and $d^E_{\min}$ to denote the maximum and minimum hyperedge sizes in $\mathcal{H}$, and define the average hyperedge size as $\bar d_e=\frac{1}{m}\sum_{e\in E}|e|$ when $m>0$. For a vertex $u\in V$, its degree $d_u$ is the number of hyperedges containing $u$. For a vertex set $S\subseteq V$, $E[S]$ denotes the set of hyperedges fully contained in $S$. Let $\delta\ge1$ be an integer and write $r_e=\min\{\delta,|e|\}$. When $E\ne\emptyset$, it suffices for a complete decomposition to examine $1\le\delta\le d^E_{\max}$, since larger values have identical head capacities. If $E=\emptyset$, every vertex has IDN zero. We write $L=\sum_{e\in E}|e|$ for the input incidence count and $R_\delta=\sum_e r_e\le m\delta$ for total allocated heads.

\subsection{Capped Hyperedge Contribution}
\begin{definition}[Capped Hyperedge Contribution]\label{def:capped-hyperedge-contribution}
For $X\subseteq V$ and an integer parameter $\delta\ge1$, with $r_e=\min\{\delta,|e|\}$, define
$F_\delta(X)=\sum_{e\in E}\max\{0,r_e-|e\setminus X|\}$.
\end{definition}
Each hyperedge $e$ contributes $r_e$ units, allocated to distinct incident vertices. Since at most $|e\setminus X|$ units can be assigned outside $X$, its minimum contribution to $X$ is $\max\{0,r_e-|e\setminus X|\}$. This bound is tight, as the contribution can be minimized by assigning as many units as possible outside $X$. Thus, $F_\delta(X)$ captures the total guaranteed contribution within $X$, independently of any particular orientation of the hypergraph. Identical hyperedges are treated as distinct occurrences and contribute separately. In particular, $F_\delta(\emptyset)=0$ and $F_\delta(V)=R_\delta$.

\begin{definition}[Capped Density]\label{density1}
For every nonempty $X\subseteq V$, define
$\rho_\delta(X)=\frac{F_\delta(X)}{|X|}$.
\end{definition}

Unlike degree-density, which considers only fully contained hyperedges, this density also accounts for the capped contributions of partially contained hyperedges that must be assigned within the vertex set.

The function $F_\delta$ is monotone and supermodular. For any $A,B\subseteq V$, we have
$F_\delta(A)+F_\delta(B)\le F_\delta(A\cup B)+F_\delta(A\cap B)$.
To see this, the contribution of each hyperedge $e$ can be expressed as $\max\{0,|e\cap X|-(|e|-r_e)\}$, which is nondecreasing and convex in $|e\cap X|$. Consequently, adding a vertex to a larger set yields a contribution gain no smaller than adding it to a subset. Since these properties hold for each hyperedge, they are preserved under summation, establishing the monotonicity and supermodularity of $F_\delta$.

\subsection{Dense Sets and Hypergraph Decomposition}
\begin{definition}[{$(k,\delta)$}-Dense Subhypergraph]\label{kdense}
For integers $k\ge0$ and $\delta\ge1$, let $Q_{k,\delta}(X)=F_\delta(X)-(k-1)|X|,
 \qquad X\subseteq V$.
The dense vertex set is
$D_{k,\delta}=\bigcap\operatorname*{arg\,max}_{X\subseteq V}Q_{k,\delta}(X)$.
The $(k,\delta)$-dense subhypergraph is $\mathcal H[D_{k,\delta}]$. We also write $D_{k,\delta}$ for its vertex set when no ambiguity arises.
\end{definition}

\begin{theorem}[Uniqueness of Dense Set]\label{thm:unique-dense-set}
For every $k\ge0$ and $\delta\ge1$, the intersection in Definition~\ref{kdense} is a maximizer of $Q_{k,\delta}$ and is its unique inclusionwise smallest maximizer. Hence, $D_{k,\delta}$ is uniquely defined.
\end{theorem}
\emph{Proof.} A maximizer exists because $V$ is finite. Subtracting the modular function $(k-1)|X|$ preserves supermodularity. If $A$ and $B$ both attain the maximum $q$, then
$2q\le Q_{k,\delta}(A\cup B)+Q_{k,\delta}(A\cap B)\le2q$.
Both their union and intersection are maximizers. Repeated intersection over the finite maximizing family gives a maximizer contained in every other one, proving the claim. \hfill$\square$

The uniqueness criterion in Theorem~\ref{thm:unique-dense-set} resolves ties at integer density levels by selecting the inclusionwise smallest maximizer. In particular, $D_{0,\delta}=V$, whereas an isolated vertex belongs to no $D_{k,\delta}$ with $k\ge1$. For $k=1$, adding an isolate leaves the objective unchanged, and the intersection excludes it. If $K=\max_v d_v$, with $K=0$ for an empty vertex set, then $F_\delta(X)\le K|X|$ and $D_{K+1,\delta}=\emptyset$.

\begin{theorem}[Nested dense sets]\label{hierarchical}
For any $\delta\ge1$ and integers $k^+\ge k\ge0$, $D_{k^+,\delta}\subseteq D_{k,\delta}$.
\end{theorem}
\emph{Proof.} Equality of the two level parameters is immediate. Otherwise let $A=D_{k,\delta}$, $B=D_{k^+,\delta}$ and $Z=B\setminus A$. Optimality against $A\cup B$ and $A\cap B$ gives
\begin{align*}
 F_\delta(A\cup B)-F_\delta(A)&\le(k-1)|Z|,\\
 F_\delta(B)-F_\delta(A\cap B)&\ge(k^+-1)|Z|.
\end{align*}
Supermodularity bounds the second left side by the first. Since $k^+>k$, these inequalities force $Z=\emptyset$. \hfill$\square$

\begin{definition}[Density Decomposition and IDN]\label{def:decomposition}
For a given $\delta$, define the indexed layers and integral dense numbers by
$L_{k,\delta}=D_{k,\delta}\setminus D_{k+1,\delta}$,
 $\bar r_v^\delta=$max$\{k\ge0:v\in D_{k,\delta}\}$.
The density decomposition at $\delta$ is the collection of nonempty indexed layers $(k,L_{k,\delta})$, equivalently the complete IDN array. Empty integer levels retain their indices. Computing the full parameter family repeats this definition for $1\le\delta\le d^E_{\max}$ when $E\ne\emptyset$.
\end{definition}
Nesting, $D_{0,\delta}=V$ and $D_{K+1,\delta}=\emptyset$ make the nonempty layers a partition of $V$. Isolates have IDN zero. If $E=\emptyset$, every IDN is zero for every positive $\delta$.

\subsection{Density Properties}

\begin{theorem}[Removal and Addition Bounds]
\label{in-outsidedense}
A set $D\subseteq V$ equals $D_{k,\delta}$ if and only if,
for every nonempty $X\subseteq D$ and every
$Y\subseteq V\setminus D$,
\begin{align*}
F_\delta(D)-F_\delta(D\setminus X)&>(k-1)|X|,\\
F_\delta(D\cup Y)-F_\delta(D)&\le(k-1)|Y|.
\end{align*}
\end{theorem}

\emph{Proof.}
We first prove necessity. Since $D_{k,\delta}$ is the
unique inclusionwise smallest maximizer of $Q_{k,\delta}$,
every proper subset $D\setminus X$ has a strictly smaller
objective value. Thus,
$Q_{k,\delta}(D)>Q_{k,\delta}(D\setminus X)$,
which gives the strict removal inequality.
Similarly, the optimality of $D$ implies
$Q_{k,\delta}(D)\ge Q_{k,\delta}(D\cup Y)$,
yielding the addition inequality.
For sufficiency, suppose $D$ satisfies both inequalities.
Let $Z\subseteq V$ be arbitrary and write $Q=Q_{k,\delta}$.
By supermodularity,
$Q(D)+Q(Z)
\le Q(D\cup Z)+Q(D\cap Z)$.
The addition inequality gives $Q(D\cup Z)\le Q(D)$,
while the removal inequality implies
$Q(D\cap Z)\le Q(D)$, with strict inequality whenever
$D\setminus Z\ne\emptyset$. Consequently,
$Q(Z)\le Q(D\cap Z)\le Q(D)$.
Hence, $D$ is a maximizer of $Q$. Moreover, equality
requires $D\subseteq Z$, so every maximizer contains
$D$. Therefore, $D$ is the unique inclusionwise
smallest maximizer, and $D=D_{k,\delta}$.
\hfill$\square$

%\emph{Proof.} Since $D$ is the smallest maximizer, a proper subset $D\setminus X$ cannot attain the same objective value. Expanding $Q_{k,\delta}(D)>Q_{k,\delta}(D\setminus X)$ gives the strict removal inequality. Optimality against $D\cup Y$ gives the addition inequality. \hfill$\square$

\begin{definition}[Layer Density]\label{layerdensity}
For a nonempty layer $L_{k,\delta}$, define
\begin{equation*}
 \lambda_{k,\delta}=\frac{F_\delta(D_{k,\delta})-F_\delta(D_{k+1,\delta})}{|L_{k,\delta}|}.
\end{equation*}
\end{definition}

This density measures the contribution of a layer relative to its higher-level set. It generally differs from the stand-alone density $\rho_\delta(L_{k,\delta})=F_\delta(L_{k,\delta})/|L_{k,\delta}|$, since $F_\delta$ is not additive. All active hyperedges are considered, including partially contained ones. Empty layers have undefined density, while the zero layer has density zero.

\begin{lemma}\label{pd1}
For nonempty $X\subseteq D_{k,\delta}$ and $Y\subseteq V\setminus D_{k,\delta}$, the removal and addition densities satisfy
\begin{align*}
 \frac{F_\delta(D)-F_\delta(D\setminus X)}{|X|}&>k-1,\\
 \frac{F_\delta(D\cup Y)-F_\delta(D)}{|Y|}&\le k-1,
\end{align*}
where $D=D_{k,\delta}$. These are density inequalities, not inequalities for the stand-alone density of an arbitrary subset.
\end{lemma}
\emph{Proof.} Divide the two capped hyperedge contribution inequalities by the positive set sizes. \hfill$\square$

\begin{lemma}\label{pd2}
Every nonempty layer satisfies $k-1<\lambda_{k,\delta}\le k$.
\end{lemma}
\emph{Proof.} Apply the strict removal inequality to $D_{k,\delta}$ with $X=L_{k,\delta}$, then the addition inequality to $D_{k+1,\delta}$ with $Y=L_{k,\delta}$. Divide by $|L_{k,\delta}|$. \hfill$\square$

\begin{proposition}[Direct layer evaluation]\label{prop:layer-evaluation}
For each edge occurrence, order its endpoint IDNs as $a_{e,1}\le\cdots\le a_{e,|e|}$. Then
\begin{align*}
 F_\delta(D_{k,\delta})&=\sum_{e\in E}\sum_{i=1}^{r_e}\mathbf1\{a_{e,i}\ge k\},\\
 |L_{k,\delta}|\lambda_{k,\delta}&=\sum_{e\in E}\sum_{i=1}^{r_e}\mathbf1\{a_{e,i}=k\}.
\end{align*}
Consequently $\sum_{k:L_{k,\delta}\ne\emptyset}|L_{k,\delta}|\lambda_{k,\delta}=R_\delta$.
\end{proposition}
\emph{Proof.} The number of endpoints of $e$ outside $D_{k,\delta}$ is $|\{i:a_{e,i}<k\}|$. Its contribution is therefore the number of the first $r_e$ ordered values that are at least $k$. Subtract the formula at $k+1$ to obtain the layer numerator. The sum of the numerators equals the total edge budget. \hfill$\square$

This formula computes the capped contribution density of each layer directly from the IDNs, without reconstructing an orientation. It accounts for partially contained hyperedges and repeated edge occurrences.

%This formula evaluates the capped hyperedge contribution layer densities directly from the IDNs, without reconstructing an orientation. In particular, partially contained edges and repeated occurrences must not be discarded.

\subsection{Directed Hypergraphs and Orientations}

To represent the allocation of density contributions among vertices, we introduce hyperedge orientations. An orientation of $\mathcal{H}$ transforms each undirected hyperedge $e\in E$ into a directed hyperedge $\vec e=(T_e,H_e)$, where $T_e$ and $H_e$ are the source and target sets, respectively, satisfying $T_e\cap H_e=\emptyset$ and $T_e\cup H_e=e$. The resulting directed hypergraph is denoted by $\vec{\mathcal{H}}=(V,\vec E)$.
Figure~\ref{hyperpath}(b)--(c) illustrates two orientations of the hypergraph in Figure~\ref{hyperpath}(a). In Figure~\ref{hyperpath}(b), $\vec e_4=(\{u_3,u_4,u_5\},\{u_2\})$ has source vertices $u_3,u_4,u_5$ and target vertex $u_2$. Given $\vec{\mathcal{H}}$, the indegree of a vertex $u$ is defined as
$\vec d_u(\vec{\mathcal{H}})=
\left|\{\vec e=(T_e,H_e)\in\vec E\mid u\in H_e\}\right|$.
We simply write $\vec d_u$ when the orientation is clear.

\begin{definition}[$\delta$-Orientation]
\label{deltapr}
Given a hypergraph $\mathcal{H}=(V,E)$ and its orientation $\vec{\mathcal{H}}=(V,\vec E)$, $\vec{\mathcal{H}}$ is called a $\delta$-orientation of $\mathcal{H}$ if for every directed hyperedge $\vec e=(T_e,H_e)\in\vec E$ corresponding to an undirected hyperedge $e\in E$, we have
$|H_e|=\min\{\delta,|e|\}$,
where $\delta$ is a positive integer. The allowed head set is all of $e$ when $|e|\le\delta$.
\end{definition}

\subsection{Hyperpaths and Egalitarian Orientation}

We next define hyperpaths and reversible hyperpaths, which formalize how density contributions can be transferred and balanced under a $\delta$-orientation.

\begin{definition}[Hyperpath, Reversible Hyperpath]
\label{repa}
An \emph{incidence-simple directed hyperpath} from $s$ to $t$ is an alternating path $s=u_0,\vec e_1,u_1,\dots,\vec e_\ell,u_\ell=t$ in which each $\vec e_i=(T_{e_i},H_{e_i})$ is distinct, $u_{i-1}\in T_{e_i}$, and $u_i\in H_{e_i}$. The path is \emph{reversible} if $\vec d_t-\vec d_s\ge2$.
\end{definition}
We write $s\rightsquigarrow t$ when $s$ can reach $t$. \emph{Reorienting} an incidence-simple path means swapping $u_{i-1}$ and $u_i$ between the tail and head sets of every traversed hyperedge, i.e., $H_{e_i}\leftarrow(H_{e_i}\setminus\{u_i\})\cup\{u_{i-1}\}$. It does not mean reversing all incidences of a hyperedge. Each $|H_{e_i}|=r_{e_i}$ is preserved; the indegree of $s$ increases by one, that of $t$ decreases by one, and every intermediate vertex has unchanged indegree. Consequently the sum of squared indegrees strictly decreases on a reversible path. An egalitarian orientation admits no such path.

\begin{figure}
	\centering
	\hspace{-4mm}
	\subfigure[\scriptsize{initial hypergraph}]{
		\label{init}
		\includegraphics[width=0.326\linewidth]{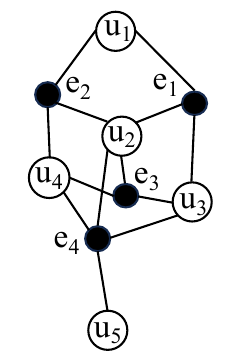}
	}
	\hspace{-6mm}
	\subfigure[\scriptsize{$1$-orientation}]{
		\label{ori}
		\includegraphics[width=0.326\linewidth]{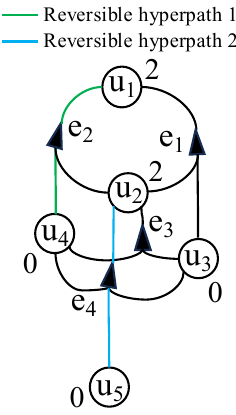}
	}
	\subfigure[\scriptsize{egalitarian $1$-orientation}]{
		\label{egaori}
		\includegraphics[width=0.326\linewidth]{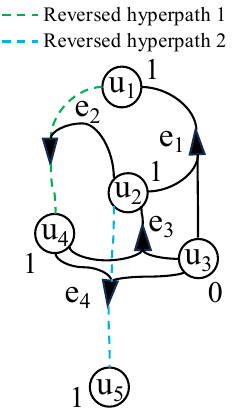}
	}

	\caption{Arbitrary and egalitarian $1$-orientations.}  
	\label{hyperpath}
\end{figure}

\begin{definition}[Egalitarian \boldmath{$\delta$}-Orientation]\label{def:egalitarian}
A $\delta$-orientation $\mathop{\mathcal{H}}\limits ^{\rightarrow}$ is said to be an \emph{egalitarian $\delta$-orientation} 
if there exists no reversible hyperpath in $\mathop{\mathcal{H}}\limits ^{\rightarrow}$.
\end{definition}

\begin{lemma}
\label{existence}
For any hypergraph $\mathcal{H} = (V, E)$, any integer parameter $\delta\ge1$, there exists an egalitarian 
$\delta$-orientation $\mathop{\mathcal{H}}\limits ^{\rightarrow}$.
\end{lemma}

\emph{Proof.}
A valid $\delta$-orientation can be constructed by selecting $r_e=\min\{\delta,|e|\}$ target vertices for each hyperedge $e\in E$ and assigning the remaining vertices to its source set. Starting from this orientation, we repeatedly reverse any reversible hyperpath. By Definition~\ref{repa}, reversing a path from $s$ to $t$ increases $\vec d_s$ by one and decreases $\vec d_t$ by one, while preserving the target-set size of every hyperedge. Hence, each reversal maintains a valid $\delta$-orientation.

To prove termination, consider the potential function
$\Psi(\vec{\mathcal{H}})=\sum_{u\in V}\vec d_u^2$.
For a reversible hyperpath from $s$ to $t$, let $a=\vec d_s$ and $b=\vec d_t$, where $b-a\ge 2$. After reversal, the change in potential is
\[
\Delta\Psi
=(a+1)^2+(b-1)^2-a^2-b^2
=-2(b-a-1)<0.
\]
Thus, every reversal strictly decreases the nonnegative integer potential $\Psi$, and the process must terminate. At termination, no reversible hyperpath remains, yielding an egalitarian $\delta$-orientation by Definition~\ref{def:egalitarian}.
\hfill$\square$

\begin{example}
Figure~\ref{ori} illustrates an arbitrary $1$-orientation containing
two reversible hyperpaths: $u_4\rightsquigarrow u_1$ through
$(u_4,\vec e_2,u_1)$ and $u_5\rightsquigarrow u_2$ through
$(u_5,\vec e_4,u_2)$.
Since $\vec d_{u_1}-\vec d_{u_4}=2$ and
$\vec d_{u_2}-\vec d_{u_5}=2$, both paths satisfy the reversal condition.
Reversing each path transfers one unit of indegree from its terminal
vertex to its starting vertex, reducing the corresponding indegree
gap by two while preserving the target-set size of each hyperedge.
The resulting orientation, shown in Figure~\ref{egaori}, contains
no reversible hyperpath and is therefore an egalitarian
$1$-orientation.
\end{example}

\subsection{Orientation Bounds and Reachability}
Orientations are a computational representation of the preceding set-function definitions. They are not needed to define $F_\delta$, $D_{k,\delta}$, the layers, or the IDNs.

For any valid $\delta$-orientation, each edge assigns at least $\max\{0,r_e-|e\setminus X|\}$ heads to $X$. Assigning heads outside $X$ first attains this lower bound independently on every edge. Therefore
 $F_\delta(X)=\min_{\text{valid }\delta\text{-orientations}}
                  \sum_{v\in X}\vec d_v$.
For a particular orientation, $F_\delta(X)\le\sum_{v\in X}\vec d_v$. Equality holds if and only if no edge has a tail outside $X$ and a head inside $X$, equivalently when $X$ is predecessor-closed. If both such incidences exist, exchanging their roles strictly reduces the assigned load in $X$; otherwise every edge already attains its minimum.

Given an egalitarian $\delta$-orientation, let $S_k=\{v:\vec d_v\ge k\}$. We define the set of vertices that can reach $S_k$ as
\[
C_{k,\delta}=\{v\in V:v\rightsquigarrow s
\text{ for some }s\in S_k\}.
\]
Reachability includes zero-length paths, so $S_k\subseteq C_{k,\delta}$.

\begin{lemma}[Reachability-set bounds]\label{degree}
No edge has a tail outside $C_{k,\delta}$ and a head inside it. Every vertex in $C_{k,\delta}$ has indegree at least $k-1$, and every vertex outside it has indegree at most $k-1$.
\end{lemma}
\emph{Proof.}
Suppose a hyperedge has a tail vertex $u\notin C_{k,\delta}$ and a head vertex $v\in C_{k,\delta}$. Since $u\rightsquigarrow v$ and $v$ can reach some seed $s\in S_k$, we have $u\rightsquigarrow s$, contradicting $u\notin C_{k,\delta}$.
Next, if a vertex $v\in C_{k,\delta}$ has $\vec d_v\le k-2$, it can reach a seed $s\in S_k$ with $\vec d_s\ge k$. Hence, $\vec d_s-\vec d_v\ge 2$, implying the existence of a reversible hyperpath and contradicting egalitarianity.
Finally, every vertex $v\notin C_{k,\delta}$ must satisfy $\vec d_v\le k-1$, since otherwise $v$ would itself be a seed in $S_k$ and thus belong to $C_{k,\delta}$. \hfill$\square$

\begin{lemma}[Dense-Set Recovery]\label{lem:kstable}
If a valid $\delta$-orientation has no path $s\rightsquigarrow t$ with $\vec d_s\le k-2$ and $\vec d_t\ge k$, then the set of vertices that can reach $\{v:\vec d_v\ge k\}$ equals $D_{k,\delta}$.
\end{lemma}
\emph{Proof.}
Let $S=\{v:\vec d_v\ge k\}$ and $C=\{v:v\rightsquigarrow s\text{ for some }s\in S\}$. By the assumed path condition, every vertex in $C$ has indegree at least $k-1$, while every vertex outside $C$ has indegree at most $k-1$. Moreover, $C$ is predecessor-closed, implying $F_\delta(C)=\sum_{v\in C}\vec d_v$.
For any nonempty $X\subseteq C$, let $B=C\setminus X$. We have
$F_\delta(C)-F_\delta(B)
=\sum_{v\in X}\vec d_v+
\left(\sum_{v\in B}\vec d_v-F_\delta(B)\right)$.
If $X$ contains a vertex in $S$, the first term exceeds $(k-1)|X|$, while the second is nonnegative. Otherwise, every vertex in $X$ has indegree $k-1$ and can reach a seed in $B$. Hence, some path from $X$ to $B$ crosses a hyperedge with a tail outside $B$ and a head inside $B$, implying $F_\delta(B)<\sum_{v\in B}\vec d_v$. Thus, the strict removal inequality holds in both cases.
For any $Y\subseteq V\setminus C$, the allocation bound and the indegrees outside $C$ give
$F_\delta(C\cup Y)-F_\delta(C)
\le \sum_{v\in Y}\vec d_v
\le (k-1)|Y|$.
Therefore, $C$ satisfies both the removal and addition bounds. By Theorem~\ref{in-outsidedense}, we conclude that $C=D_{k,\delta}$.
\hfill$\square$

\begin{theorem}[IDNs via Reachability]\label{thm:idn-reach}
For any egalitarian $\delta$-orientation, the IDN of each vertex $v$ is given by
$\bar r_v^\delta=\max_{w:\,v\rightsquigarrow w}\vec d_w$,
where reachability includes $v$ itself.
\end{theorem}

\emph{Proof.}
Consider an egalitarian $\delta$-orientation. Let $R(v)$ denote
the set of vertices reachable from $v$, including $v$ itself.
By the equivalence $D_{k,\delta}=C_{k,\delta}$, a vertex $v$
belongs to $D_{k,\delta}$ precisely when there exists some
$w\in R(v)$ with $\vec d_w\ge k$.
Let $M(v)=\max_{w\in R(v)}\vec d_w$ denote the maximum
indegree among vertices reachable from $v$. It follows that
$v\in D_{k,\delta}$ holds if and only if $k\le M(v)$.
Since the IDN of $v$ is the largest density level containing
$v$, we obtain
$\bar r_v^\delta
=\max\{k\ge 0:v\in D_{k,\delta}\}
=\max_{w\in R(v)}\vec d_w$.
Therefore, the IDN of each vertex equals its maximum
reachable indegree.
\hfill$\square$

\stitle{Remark on $\delta$.}
The parameter $\delta$ determines the contribution of each hyperedge to its incident vertices.
No hyperedge is removed: if $|e|>\delta$, then $\delta$ incident vertices receive one unit of contribution; otherwise, all incident vertices do.
Thus, a larger $\delta$ allows each hyperedge to support more vertices, whereas a smaller $\delta$ imposes a stricter contribution cap. The number of distinct occupied IDN layers need not increase monotonically with $\delta$.

\stitle{Objective and Relation to Cores.}
Our density hierarchy is defined by maximizing the penalized capped hyperedge contribution $Q_{k,\delta}$, rather than the conventional induced degree-density $\rho(X)$. The preceding orientation-based characterization provides an equivalent way to compute these dense sets.
When $\delta=1$, every vertex in $D_{k,1}$ is incident to at least $k$ hyperedges fully contained in $D_{k,1}$. Thus, $D_{k,1}$ is contained in the hypergraph $k$-core. For ordinary loopless graphs, a stronger relationship holds:
$C_{2k}\subseteq D_{k,1}\subseteq C_k$,
where $C_k$ denotes the conventional $k$-core.
To establish $D_{k,1}\subseteq C_k$ for $k\ge1$, recall that $D_{k,1}$ is predecessor-closed. Hence, every incoming edge of a vertex in $D_{k,1}$ is fully contained in this set. Each seed has indegree at least $k$ and therefore at least $k$ incident edges. Each nonseed has indegree $k-1$ and can reach a seed through a directed path. The first edge on this path is an additional incident edge contained in $D_{k,1}$. Consequently, every vertex in $D_{k,1}$ has induced degree at least $k$, proving the inclusion.
For the reverse inclusion $C_{2k}\subseteq D_{k,1}$, let $A=C_{2k}\setminus D_{k,1}$. Since $D_{k,1}$ is predecessor-closed, every edge between $A$ and $C_{2k}\cap D_{k,1}$ must be oriented toward $A$. Moreover, each edge within $A$ contributes one unit of indegree to $A$. Therefore, the total indegree allocated to $A$ is at least half its degree sum within $C_{2k}$. As every vertex in $C_{2k}$ has induced degree at least $2k$, the total allocated indegree in $A$ is at least $k|A|$. However, every vertex outside $D_{k,1}$ has indegree at most $k-1$, implying a total indegree of at most $(k-1)|A|$. This contradiction forces $A=\emptyset$. For $k=0$, both inclusions follow directly from $D_{0,1}=C_0=V$.
%These relationships connect the proposed model to classical core decomposition. Unlike degree-based cores, however, our density levels are determined by balanced hyperedge contributions, allowing finer distinctions between cohesive structures.

\subsection{Density and Boundary Guarantee}
We relate the capped hyperedge contribution to fully-contained-edge degree-density. %These quantities need not coincide.

\begin{definition}[Degree-Density]\label{adensity}
For nonempty $X\subseteq V$, let $d_v(X)=|\{e\in E[X]:v\in e\}|$ and define
\[
 \rho(X)=\frac{1}{|X|}\sum_{v\in X}d_v(X)
        =\frac{\sum_{e\in E[X]}|e|}{|X|}.
\]
\end{definition}

\begin{definition}[Capped Hyperedge Contribution Internalization Coefficient]\label{def:theta}
For a nonempty $X\subseteq V$ with $F_\delta(X)>0$, define
$\theta(X)=\frac{\sum_{e\in E[X]}r_e}{F_\delta(X)}$.
Set $\theta(X)=0$ when $F_\delta(X)=0$. The coefficient depends only on the hypergraph, $X$ and $\delta$, and lies in $[0,1]$.
\end{definition}

Every hyperedge fully contained in $X$ contributes its entire budget $r_e$ to $F_\delta(X)$, while partially contained hyperedges may contribute additional nonnegative amounts. Therefore, the total contribution from internal hyperedges cannot exceed $F_\delta(X)$, yielding
$0\le \sum_{e\in E[X]}r_e\le F_\delta(X)$.
It follows that $0\le\theta(X)\le1$ whenever $F_\delta(X)>0$. The same bound holds when $F_\delta(X)=0$ by definition.

\begin{lemma}\label{lem:headfloor-Dk}
For every nonempty $D=D_{k,\delta}$,
\[
 \rho_\delta(D)\ge(k-1)+\frac{1}{|D|}.
\]
\end{lemma}
\emph{Proof.}
Let $D=D_{k,\delta}$ be a nonempty dense set.
Applying the strict removal inequality with $X=D$, we obtain
$F_\delta(D)-F_\delta(\emptyset)>(k-1)|D|$.
Since $F_\delta(\emptyset)=0$, this reduces to
$F_\delta(D)>(k-1)|D|$.
Moreover, both $F_\delta(D)$ and $(k-1)|D|$ are integers,
so the strict inequality implies
$F_\delta(D)\ge (k-1)|D|+1$.
Dividing both sides by $|D|>0$ yields
$\rho_\delta(D)=\frac{F_\delta(D)}{|D|}
\ge k-1+\frac{1}{|D|}$,
which establishes the claimed lower bound.
\hfill$\square$

\begin{lemma}\label{lem:bridge-strong}
For every nonempty $X\subseteq V$, $\rho(X)\ge\theta(X)\rho_\delta(X)$.
\end{lemma}

\emph{Proof.}
Let $X\subseteq V$ be nonempty. If $F_\delta(X)=0$,
then $\theta(X)=0$ by definition, and the inequality holds
since $\rho(X)\ge 0$.
Otherwise, recall that $r_e=\min\{\delta,|e|\}\le |e|$
for every hyperedge $e$. Summing over all hyperedges
fully contained in $X$, we obtain
$\sum_{e\in E[X]}|e|
\ge \sum_{e\in E[X]}r_e
=\theta(X)F_\delta(X)$.
Dividing both sides by $|X|$ and applying the definitions
of degree-density and capped density gives
$\rho(X)
=\frac{\sum_{e\in E[X]}|e|}{|X|}
\ge \theta(X)\frac{F_\delta(X)}{|X|}
=\theta(X)\rho_\delta(X)$.
This establishes the desired inequality.
\hfill$\square$

\begin{theorem}[Capped Hyperedge Density Guarantee]\label{thm:density-Dk}
For every nonempty $D=D_{k,\delta}$,
\[
 \rho(D)\ge\theta(D)\rho_\delta(D)
       \ge\theta(D)\left(k-1+\frac{1}{|D|}\right).
\]
\end{theorem}
\emph{Proof.} Combine Lemmas~\ref{lem:headfloor-Dk} and~\ref{lem:bridge-strong}, using $\theta(D)\ge0$. \hfill$\square$

\begin{definition}[Boundary-to-Volume Ratio]\label{def:one-sided}
For $X\subseteq V$ with $\mathrm{vol}(X)=\sum_{v\in X}d_v>0$, let $b(X)$ be the number of hyperedges meeting both $X$ and $V\setminus X$, and define $\phi_{\mathrm{out}}(X)=b(X)/\mathrm{vol}(X)$. This is a one-sided boundary ratio, not the standard symmetric conductance.
\end{definition}
\begin{theorem}[Boundary Ratio Bound]\label{thm:cond}
For every nonempty $X$ with $\mathrm{vol}(X)>0$, writing $\bar d(X)=\mathrm{vol}(X)/|X|$,
\[
 \phi_{\mathrm{out}}(X)\le 1-\frac{\rho(X)}{\bar d(X)}.
\]
In particular, if $D_{k,\delta}\ne\emptyset$ and $\mathrm{vol}(D_{k,\delta})>0$, then
\[
 \phi_{\mathrm{out}}(D_{k,\delta})\le
 1-\frac{\theta(D_{k,\delta})\rho_\delta(D_{k,\delta})}{\bar d(D_{k,\delta})}.
\]
\end{theorem}
\emph{Proof.}
Consider any nonempty vertex set $X\subseteq V$ with $\mathrm{vol}(X)>0$. The incidences within $X$ can be divided into those belonging to fully contained hyperedges and those belonging to crossing hyperedges. The former contribute $\sum_{e\in E[X]}|e|=\rho(X)|X|$, while each crossing hyperedge contributes at least one incidence to $X$. Therefore,
$b(X)\le \mathrm{vol}(X)-\rho(X)|X|$.
Dividing both sides by $\mathrm{vol}(X)=\bar d(X)|X|$, we obtain
$\phi_{\mathrm{out}}(X)
=\frac{b(X)}{\mathrm{vol}(X)}
\le 1-\frac{\rho(X)}{\bar d(X)}$.
For a nonempty dense set $D_{k,\delta}$ with positive volume, Theorem~\ref{thm:density-Dk} gives
$\rho(D_{k,\delta})\ge\theta(D_{k,\delta})\rho_\delta(D_{k,\delta})$.
Substituting this bound into the inequality above yields the result.
\hfill$\square$

The symmetric ratio $b(X)/\min\{\mathrm{vol}(X),\mathrm{vol}(V\setminus X)\}$ is not bounded this way unless the denominator is $\mathrm{vol}(X)$ (or the corresponding volume-side condition is imposed).

\subsection{Search and Decomposition}
Given a hypergraph, a nonnegative density level $k$ and a positive integer $\delta$, the search task returns $D_{k,\delta}$. The decomposition task takes $\delta$ as input and computes the canonical IDN array, hence every nested set $D_{k,\delta}$. Computing all effective parameter values repeats the procedure for $1\le\delta\le d^E_{\max}$. The challenge is to balance contributions without enumerating orientations and recover the entire hierarchy without repeating a search for every $k$.

\section{Dense Subhypergraph Search}
\label{sec:search}
Let $L=\sum_{e\in E}|e|$ and $R_\delta=\sum_{e\in E}\min\{\delta,|e|\}\le m\delta$. The directed incidence representation contains $O(n+m+L)$ vertices and arcs. Every exchange below is performed on an incidence-simple hyperpath by swapping one tail and one head of each traversed hyperedge; exchanging entire hyperedge head and tail sets would not preserve the $\delta$-orientation.

\subsection{The Path-Based Algorithm: $\kw{DSM\text{-}PATH}$}
\kw{DSM\text{-}PATH} constructs an egalitarian $\delta$-orientation by repeatedly identifying and reversing reversible hyperpaths. Starting from an arbitrary $\delta$-orientation, we first construct a directed incidence graph by introducing a virtual vertex $z_e$ for each hyperedge $e$. For every $u\in T_e$ and $v\in H_e$, we add directed arcs $u\rightarrow z_e$ and $z_e\rightarrow v$, respectively. This representation preserves reachability in the directed hypergraph.
To find a reversible hyperpath, we compute the strongly connected components (SCCs) of the incidence graph and contract each component into a single vertex, obtaining a directed acyclic graph (DAG). Each component is assigned the maximum indegree of its original vertices. By propagating these values in reverse topological order, we obtain the maximum reachable indegree $M(v)$ for every vertex $v$. If some vertex $s$ satisfies $M(s)\ge\vec d_s+2$, a reversible hyperpath exists. We then use breadth-first search (BFS) to find a path from $s$ to a reachable vertex $t$ with $\vec d_t\ge\vec d_s+2$.

Reversing this path transfers one unit of indegree from $t$ to $s$ while preserving the target-set size of every traversed hyperedge. The algorithm repeats this process until no reversible hyperpath remains, yielding an egalitarian $\delta$-orientation. Each round of SCC computation, reachability propagation, and path reconstruction takes $O(n+m+L)$ time, where $L=\sum_{e\in E}|e|$.

\begin{algorithm}[t]
\small
\caption{$\kw{DSM\text{-}PATH}(\mathcal H,k,\delta)$}\label{dshs}
\KwData{$\mathcal H=(V,E)$, integers $k\ge0$, $\delta\ge1$}
\KwResult{$D_{k,\delta}$}
Construct any valid $\delta$-orientation $\vec{\mathcal H}$\;
\While{some incidence-simple reversible hyperpath $s\rightsquigarrow t$ exists}{
 Swap one tail/head pair along every traversed hyperedge\;
}
$S\gets\{v:\vec d_v\ge k\}$\;
\Return{$\{v:v\rightsquigarrow s\text{ for some }s\in S\}$}\;
\end{algorithm}
\begin{theorem}[Correctness and a conservative PATH bound]\label{comdshs}
Algorithm~\ref{dshs} returns $D_{k,\delta}$. With SCC-based witness detection, its time and space bounds are $O((R_\delta^2+1)(n+m+L))$ and $O(n+m+L)$.
\end{theorem}

\emph{Proof.}
Each path reversal preserves the target-set size of every hyperedge and changes only the indegrees of its endpoints. For a reversible hyperpath from $s$ to $t$, the indegree of $s$ increases by one while that of $t$ decreases by one, preserving a valid $\delta$-orientation.
Consider the potential function $\Psi=\sum_{v\in V}\vec d_v^2$. Since $\vec d_t-\vec d_s\ge2$, each reversal decreases the potential by
$-\Delta\Psi=2(\vec d_t-\vec d_s-1)\ge2$.
As $0\le\Psi\le R_\delta^2$, at most $O(R_\delta^2)$ reversals are performed. Each iteration uses SCC-based reachability analysis and BFS to identify and reconstruct a reversible hyperpath, taking $O(n+m+L)$ time. Thus, the overall time complexity is $O((R_\delta^2+1)(n+m+L))$, with $O(n+m+L)$ space.
Upon termination, no reversible hyperpath remains, yielding an egalitarian $\delta$-orientation. By Lemma~\ref{lem:kstable}, the vertices that can reach a seed of indegree at least $k$ form $D_{k,\delta}$, proving correctness.
\hfill$\square$

\subsection{The Flow-Based Algorithm: $\kw{DSM\text{-}FLOW}$}
\label{subsec:dsm-flow}

\kw{DSM\text{-}PATH} computes a globally egalitarian orientation, but finding $D_{k,\delta}$ requires only eliminating directed hyperpaths from vertices with indegree at most $k-2$ to those with indegree at least $k$. Reversing such a path increases the indegree of its starting vertex by one and decreases that of its final vertex by one. Instead of processing these paths individually, \kw{DSM\text{-}FLOW} uses maximum flow to perform multiple reorientations in a single computation.

\stitle{Reorientation flow network.}
Given a valid $\delta$-orientation, we construct a directed flow network
$\mathcal N=(V_N,A,c)$ with a source $s$ and a sink $t$.
Each hyperedge $e$ is represented by a factor vertex $z_e$.
A directed two-arc segment $u\to z_e\to v$, with $u\in T_e$ and
$v\in H_e$, represents a possible exchange of a head from $v$ to $u$.
We use $q=k-1$ as the pivot separating vertices that can receive
additional contributions from vertices that can release them.

\begin{definition}[Reorientation Flow Network]\label{flow}
For a valid $\delta$-orientation and $q=k-1$, let
$V_N=V\cup\{z_e:e\in E\}\cup\{s,t\}$. The arcs and their capacities
are defined as follows:
\begin{itemize}[leftmargin=12pt,topsep=2pt,itemsep=1pt]
\item For each $e\in E$, add $u\to z_e$ for every $u\in T_e$ and
      $z_e\to v$ for every $v\in H_e$, each with unit capacity.
\item For each $u$ with $\vec d_u<q$, add $s\to u$ with capacity
      $q-\vec d_u$.
\item For each $v$ with $\vec d_v>q$, add $v\to t$ with capacity
      $\vec d_v-q$.
\end{itemize}
The network has $n+m+2$ vertices and $O(n+L)$ arcs, where
$L=\sum_{e\in E}|e|$.
\end{definition}

The source and sink capacities bound the number of contributions that
low- and high-indegree vertices may exchange. In particular, every
$s$--$t$ augmenting path corresponds to a directed hyperpath from a
vertex below $q$ to a vertex above $q$. Sending one unit of flow
along this path transfers one head contribution in the opposite
direction, while preserving the head budget of each traversed hyperedge.
Because all hyperedge-incidence arcs have unit capacity, an integral
maximum flow can be converted back to a valid $\delta$-orientation.

\begin{figure}[t]
\centering
\subfigure[\scriptsize{$1$-orientation}]{
  \label{fig3a}
  \includegraphics[width=0.297\linewidth]{fig_hyperpath1.pdf}
}
\subfigure[\scriptsize{Reorientation flow network}]{
  \label{fig3b}
  \includegraphics[width=0.297\columnwidth]{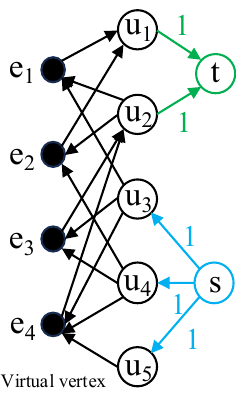}
}
\subfigure[\scriptsize{Orientation after flow}]{
  \label{fig3c}
  \includegraphics[width=0.297\columnwidth]{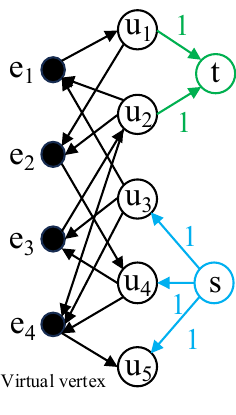}
}
\caption{Reorientation of a hypergraph using maximum flow for $k=2$ and
$\delta=1$.}
\label{flownetwork}
\end{figure}

\begin{example}
Consider the $1$-orientation in Figure~\ref{flownetwork}(a), where the
indegrees of $u_1,\ldots,u_5$ are $(2,2,0,0,0)$.
For $k=2$, the pivot is $q=1$. In Figure~\ref{flownetwork}(b),
the source $s$ is connected to $u_3,u_4,u_5$, each with capacity one,
and $u_1,u_2$ are connected to the sink $t$, also with unit capacity.
For each hyperedge $e_i$, the factor vertex $z_{e_i}$ connects its
tail vertices to its head vertices through unit-capacity arcs.
For example, the network contains the two augmenting paths
$s\to u_4\to z_{e_2}\to u_1\to t, s\to u_5\to z_{e_4}\to u_2\to t$.
They carry two units of flow in total, which is maximum because the
total capacity into $t$ is two. The first path replaces $u_1$ by $u_4$
in the head set of $e_2$, while the second replaces $u_2$ by $u_5$
in the head set of $e_4$. The resulting indegrees are $(1,1,0,1,1)$,
as illustrated in Figure~\ref{flownetwork}(c). No vertex now has
indegree at least two, so $S_2=\emptyset$ and $D_{2,1}=\emptyset$.
\end{example}

\begin{algorithm}[t]
\small
\caption{$\kw{DSM\text{-}FLOW}(\mathcal H,k,\delta)$}\label{dshs2}
\KwData{$\mathcal H=(V,E)$, integers $k\ge0$, $\delta\ge1$}
\KwResult{$D_{k,\delta}$}
\If{$k=0$}{\Return{$V$}\;}
Construct any valid $\delta$-orientation $\vec{\mathcal H}$\;
$q\gets k-1$, $V_N\gets V\cup\{z_e:e\in E\}\cup\{s,t\}$,
$A\gets\emptyset$\;
\ForEach{$e\in E$}{
  \ForEach{$u\in T_e$}{
    Add $u\to z_e$ to $A$ with capacity $1$\;
  }
  \ForEach{$v\in H_e$}{
    Add $z_e\to v$ to $A$ with capacity $1$\;
  }
}
\ForEach{$v\in V$}{
  \If{$\vec d_v<q$}{
    Add $s\to v$ to $A$ with capacity $q-\vec d_v$\;
  }
  \If{$\vec d_v>q$}{
    Add $v\to t$ to $A$ with capacity $\vec d_v-q$\;
  }
}
$f\gets\operatorname{MaxFlow}(V_N,A,c,s,t)$\;
\ForEach{$e\in E$}{
  $U_e\gets\{u\in T_e:f(u,z_e)=1\}$\;
  $W_e\gets\{v\in H_e:f(z_e,v)=1\}$\;
  $H_e\gets(H_e\setminus W_e)\cup U_e$;
  $T_e\gets e\setminus H_e$\;
}
Recompute $\vec d_v$ from the updated head sets\;
$S\gets\{v\in V:\vec d_v\ge k\}$\;
$D_{k,\delta}\gets$ vertices reached by reverse traversal from $S$
in the directed incidence graph\;
\Return{$D_{k,\delta}$}\;
\end{algorithm}

Algorithm~\ref{dshs2} first constructs a valid $\delta$-orientation
and initializes the flow network at pivot $q=k-1$.
The hyperedge arcs encode feasible tail-to-head exchanges, while the
source and sink arcs specify how many contributions each vertex can
receive or release. After computing an integral maximum flow, the
algorithm identifies, for each $e$, the tail incidences $U_e$ and
head incidences $W_e$ carrying one unit of flow.
Flow conservation at $z_e$ ensures $|U_e|=|W_e|$; replacing $W_e$
with $U_e$ therefore preserves $|H_e|=r_e$.
This simultaneous update is equivalent to pairing the used tail and
head incidences of each hyperedge, without reversing all incidences of
the hyperedge. Finally, a reverse traversal from the vertices with
indegree at least $k$ returns all vertices that can reach such a seed.
%Unlike \kw{DSM\text{-}PATH}, \kw{DSM\text{-}FLOW} does not require global egalitarianity.

\begin{theorem}[Correctness of $\kw{DSM\text{-}FLOW}$]\label{comflowcorr}
Algorithm~\ref{dshs2} returns $D_{k,\delta}$.
\end{theorem}
\emph{Proof.}
The flow is integral, and conservation at every factor vertex
$z_e$ gives $|U_e|=|W_e|$. The updated head set therefore contains
$r_e$ vertices, so the recovered orientation is valid.
For each original vertex $v$, its new indegree is
$\vec d_v'=\vec d_v+f(s,v)-f(v,t)$, where absent terminal arcs
carry zero flow. Thus a vertex below $q$ can only gain contributions,
and a vertex above $q$ can only lose them, without crossing $q$.
After applying the flow, the residual incidence arcs represent
precisely the tail-to-head exchanges allowed by the updated
orientation. If a vertex of final indegree at most $k-2$ could reach
another vertex of final indegree at least $k$, the residual source and
sink capacities would extend that directed path to an augmenting
$s$--$t$ path. This contradicts the maximality of $f$.
Hence the final orientation satisfies Lemma~\ref{lem:kstable},
and the reverse traversal in Algorithm~\ref{dshs2} returns
$D_{k,\delta}$.
\hfill$\square$

\begin{theorem}[Conservative Flow Complexity]\label{comflow}
With a generic Dinic max-flow algorithm, Algorithm~\ref{dshs2}
uses $O((n+m)^2(n+L))$ time and $O(n+m+L)$ space.
\end{theorem}
\emph{Proof.}
The network has $O(n+m)$ vertices and $O(n+L)$ arcs. The generic
$O(|V_N|^2|A|)$ bound for Dinic dominates network construction,
flow recovery, and the final reverse traversal, which all take
$O(n+m+L)$ time. The network and auxiliary arrays require
$O(n+m+L)$ space.
\hfill$\square$

\subsection{Greedy Initialization: $\kw{DSM\text{-}FLOW+}$}
\kw{DSM\text{-}FLOW} constructs an arbitrary $\delta$-orientation, which may concentrate hyperedge contributions on vertices with high indegrees and increase the subsequent flow computation. \kw{DSM\text{-}FLOW+} improves this initialization through a greedy head assignment strategy.
Specifically, hyperedges are processed sequentially. For each hyperedge $e$, \kw{DSM\text{-}FLOW+} selects $r_e=\min\{\delta,|e|\}$ incident vertices with the smallest current indegrees as its head set $H_e$, breaking ties by vertex ID. The indegrees of the selected vertices are updated immediately before processing the next hyperedge. This strategy considers previously assigned contributions and reduces the initial imbalance among vertices. After constructing the $\delta$-orientation, \kw{DSM\text{-}FLOW+} follows the same maximum-flow and reachability procedures as \kw{DSM\text{-}FLOW} to compute $D_{k,\delta}$. 

\begin{algorithm}[t]
\small
\caption{$\kw{DSM\text{-}FLOW+}(\mathcal H,k,\delta)$}\label{dshs3}
\KwData{$\mathcal H=(V,E)$, integers $k\ge0$, $\delta\ge1$}
\KwResult{$D_{k,\delta}$}
\If{$k=0$}{\Return{$V$}\;}
Initialize every indegree to zero\;
\ForEach{hyperedge $e\in E$}{
 Select $r_e=\min\{\delta,|e|\}$ incident vertices of lowest current indegree\;
 Assign these $r_e$ vertices to $H_e$, update their indegrees\;
}
Run the flow and reachability steps of Algorithm~\ref{dshs2}\;
\Return{$D_{k,\delta}$}\;
\end{algorithm}
\begin{theorem}[Complexity of $\kw{DSM\text{-}FLOW+}$]\label{comflowplus}
Algorithm~\ref{dshs3} is correct and has the same conservative asymptotic bounds as Algorithm~\ref{dshs2}, plus $O(\sum_e |e|\log |e|)$ time for sorting and choosing a prefix; a linear-time alternative is possible.
\end{theorem}
\emph{Proof.} Greedy initialization creates a valid $\delta$-orientation. Algorithm~\ref{dshs2} is correct from any valid initialization. \hfill$\square$

\section{Density Decomposition Algorithms}
\label{sec:two-algorithms}

The algorithms in Section~3 compute $D_{k,\delta}$ for a given density level $k$. Applying them independently to every level would repeat substantial computation. Instead, we compute the complete IDN array from an egalitarian $\delta$-orientation, where each dense set can be recovered as $D_{k,\delta}=\{v\in V:\bar r_v^\delta\ge k\}$. This avoids repeated searches and the explicit storage of nested dense sets.

We develop two algorithms, \kw{Flow\text{-}Decomp} and \kw{Reverse\text{-}Decomp}, both using greedy initialization, integral flow balancing, and residual-network reuse. \kw{Flow\text{-}Decomp} computes maximum reachable indegrees through SCC condensation, while \kw{Reverse\text{-}Decomp} uses descending-load reverse propagation. Once the orientation is balanced, both recover all IDNs in $O(n+m+L)$ time. For different values of $\delta$, the decomposition can be computed independently over $\delta=1,\ldots,d^E_{\max}$.

\subsection{Integral Flow Balancing}
\label{subsec:reusable-flow}
We construct an initial $\delta$-orientation using the greedy strategy of DSM-FLOW+. Starting with zero indegrees, we process each hyperedge $e$ by selecting the $r_e=\min\{\delta,|e|\}$ incident vertices with the smallest current indegrees as its head set, breaking ties by vertex ID. Their indegrees are updated before processing the next hyperedge. This procedure takes $O(\sum_e |e|\log(|e|+1))$ time using sorting and produces a valid $\delta$-orientation, although further balancing may be required to obtain an egalitarian orientation.

For the current orientation, let $M(v)=\max_{t:v\rightsquigarrow t}\vec d_t$, including $t=v$. If some vertex $s$ satisfies $M(s)\ge\vec d_s+2$, a reversible hyperpath exists. We select an integer pivot $q$ strictly between $\vec d_s$ and $M(s)$ and apply the reorientation flow network in Definition~\ref{flow}. The source and sink capacities are $\max\{0,q-\vec d_v\}$ and $\max\{0,\vec d_v-q\}$, respectively. Hyperedges with $|e|\le\delta$ have no tails and require no movable factor nodes, while their head contributions remain in the vertex indegrees. The network contains $O(n+m)$ vertices and $O(n+L)$ arcs.

Each unit of integral flow changes the head assignments along a directed hyperpath, increasing the indegree of a lower-load vertex and decreasing that of a higher-load vertex. Flow conservation preserves the head budget of every hyperedge. Consequently, each unit of flow decreases the potential $\Psi=\sum_{v\in V}\vec d_v^2$ by at least two. The process terminates when $M(v)\le\vec d_v+1$ for all vertices, yielding an egalitarian $\delta$-orientation.
To avoid rebuilding the network in every iteration, we retain the residual capacities of the incidence arcs, which represent the updated orientation. Only the source and sink arcs are reset for the next pivot. The network is constructed when a reversible hyperpath is first detected and reused in subsequent balancing rounds.

\begin{theorem}[Residual Network Reuse]\label{thm:reuse-flow}
After each integral pivot flow, the incidence residual network represents a valid $\delta$-orientation. Reusing this network in subsequent balancing rounds is equivalent to rebuilding it from the updated orientation and preserves the computed IDNs.
\end{theorem}

\emph{Proof.}
Each incidence is represented by two opposite residual arcs with complementary binary capacities. An integral flow updates these capacities without changing their sum. By flow conservation at each factor vertex $z_e$, the number of assigned heads remains $r_e$. Thus, the updated residual network represents a valid $\delta$-orientation.
In subsequent balancing rounds, the incidence residual arcs are retained, while the source and sink arcs are reset according to the new pivot and current indegrees. Since the retained arcs encode the updated head assignments, the resulting network is equivalent to one constructed from the current orientation. This equivalence holds throughout the balancing process. Upon reaching an egalitarian orientation, Theorem~\ref{thm:idn-reach} ensures that both algorithms recover the correct IDNs.
\hfill$\square$

\subsection{\kw{Flow\text{-}Decomp}: SCC-Based Decomposition}
\label{subsec:scc-all}
Flow-Decomp computes the maximum reachable indegree of each vertex using strongly connected component (SCC) condensation. For the current $\delta$-orientation, we first construct a directed incidence graph by introducing a factor vertex $z_e$ for each hyperedge $e\in E$. For every $u\in T_e$ and $v\in H_e$, we add directed arcs $u\to z_e$ and $z_e\to v$, respectively. These arcs preserve reachability in the directed hypergraph.

We then compute the SCCs of the incidence graph and contract each SCC into a single vertex, forming a directed acyclic graph (DAG). For each component $C$, let $h(C)$ be the maximum indegree of its original vertices, or $-1$ if $C$ contains only factor vertices. Since vertices within the same SCC have identical reachable components, their maximum reachable indegrees can be computed together. Specifically, we process the DAG in reverse topological order and compute
$M(C)=\max\left\{h(C),\max_{C'\in N^+(C)}M(C')\right\}$,
where $N^+(C)$ denotes the outgoing neighbors of $C$ in the DAG. Each original vertex $v$ belonging to $C$ is then assigned $M(v)=M(C)$.
If some vertex $v$ satisfies $M(v)\ge\vec d_v+2$, a reversible hyperpath exists, and Flow-Decomp selects a vertex with the largest indegree gap to perform flow balancing. The SCCs and reachable maxima are recomputed after each balancing round. When $M(v)\le\vec d_v+1$ holds for all vertices, no reversible hyperpath remains and the orientation is egalitarian. By Theorem~\ref{thm:idn-reach}, the computed values $M(v)$ are the IDNs. Since the incidence graph contains $O(n+m)$ vertices and $O(L)$ arcs, SCC computation and DAG propagation take $O(n+m+L)$ time per round.

\begin{algorithm}[t]
\small
\caption{$\kw{Flow\text{-}Decomp}(\mathcal H,\delta)$}\label{alg:scc-all}
\KwData{Normalized hypergraph $\mathcal H$, integer $\delta\ge1$}
\KwResult{Complete IDN array $\bar r^\delta$}
Construct a greedy valid $\delta$-orientation\;
\While{true}{
 Build the directed incidence graph and its SCC condensation\;
 Compute reachable maxima $M$ by reverse-topological propagation\;
 \If{$M(v)\le\vec d_v+1$ for every $v$}{\Return{$M$}\;}
 Select $s$ with the largest reachable load gap $M(s)-\vec d_s\ge2$\;
Set $k\gets\max\{\vec d_s+2,\vec d_s+\lceil(M(s)-\vec d_s)/2\rceil\}$ and $q\gets k-1$\;
 Apply an integral maximum flow at pivot $q$, retaining incidence residual pairs\;
}
\end{algorithm}

\begin{theorem}[Correctness of \kw{Flow\text{-}Decomp} ]\label{thm:scc-all}
Algorithm~\ref{alg:scc-all} returns the canonical IDNs at $\delta$. Its reachability computation costs $O(n+m+L)$ per round.
\end{theorem}
\emph{Proof.} Paths in the factor graph correspond to directed hyperpaths. Vertices in an SCC have the same reachable components, and DAG propagation computes their largest reachable indegree. Every positive flow preserves legality and strictly decreases $\Psi$. The stopping condition excludes every reversible hyperpath. Thus the final orientation is egalitarian, and $M$ equals the IDN array by Theorem~\ref{thm:idn-reach}. SCCs and the condensation have total size $O(n+m+L)$. \hfill$\square$

\subsection{\kw{Reverse\text{-}Decomp}: Direct Reverse Propagation}
\label{subsec:rr-all}
For each movable edge, a head-to-edge CSR index supports reverse traversal from its heads to its tails. Process potential terminals in nonincreasing indegree order. Starting from an unlabeled terminal $t$, visit its unlabeled reverse predecessors and assign them $\vec d_t$. Mark each vertex and each factor only at its first visit. Edges with no tails can be omitted from this traversal while retaining their indegree contributions.

\begin{algorithm}[t]
\small
\caption{$\kw{Reverse\text{-}Decomp}(\mathcal H,\delta)$}\label{alg:rr-reach}
\KwData{Normalized hypergraph $\mathcal H$, integer $\delta\ge1$}
\KwResult{Complete IDN array $\bar r^\delta$}
Construct a greedy valid $\delta$-orientation\;
\While{true}{
 Build a head-to-edge CSR index for movable edges\;
 Set $M(v)\gets-1$; mark all factors unvisited\;
 \ForEach{terminal $t$ in nonincreasing indegree order}{
  \If{$M(t)=-1$}{
   Reverse-traverse from $t$, visiting each factor and vertex at most once; label reached vertices with $\vec d_t$\;
  }
 }
 \If{$M(v)\le\vec d_v+1$ for every $v$}{\Return{$M$}\;}
Select $s$ with the largest reachable load gap; choose $k$ and $q=k-1$ as in Algorithm~\ref{alg:scc-all}\;
 Apply an integral pivot flow; synchronize changed head incidences and retain residual pairs\;
}
\end{algorithm}

\begin{theorem}[Correctness of \kw{Reverse\text{-}Decomp}]\label{thm:rr-all}
For every valid intermediate orientation, the reverse pass in Algorithm~\ref{alg:rr-reach} computes $M(v)=\max_{t:v\rightsquigarrow t}\vec d_t$ in $O(n+m+L)$ time and space. The complete algorithm returns the canonical IDNs.
\end{theorem}
\emph{Proof.} Reverse traversal from $t$ visits precisely vertices that can reach $t$. When a vertex is first labeled from $t$, any reachable terminal of larger indegree has already been processed and would have labeled that vertex earlier. Thus its first label is its maximum reachable indegree. Reusing a visited factor cannot improve a label because later terminals have no larger load. Every vertex and factor is visited once. Since $0\le\vec d_v\le m$, counting buckets order terminals in $O(n+m)$ time; scanning the CSR and tails costs $O(L)$. The resulting witness and stopping condition are identical to those of \kw{Flow\text{-}Decomp}, so the shared balancing argument and Theorem~\ref{thm:idn-reach} prove correctness. \hfill$\square$

\subsection{Complexity Analysis}

\begin{theorem}[Complexity of Density Decomposition]
\label{thm:decomp-complexity}
Let $B=n+m+L$, where $L=\sum_{e\in E}|e|$.
For either Flow-Decomp or Reverse-Decomp, let $q$ denote
the number of positive flow-balancing rounds, and let
$T_{\rm flow}^{(i)}$ be the maximum-flow computation time
in round $i$. Both algorithms require
\[
O\left(
\sum_{e\in E}|e|\log(|e|+1)
+(q+1)B
+\sum_{i=1}^{q}T_{\rm flow}^{(i)}
\right)
\]
time and $O(B)$ space.
Moreover, $q\le R_\delta^2/2$, where
$R_\delta=\sum_{e\in E}\min\{\delta,|e|\}$.
Using Dinic's algorithm,
$T_{\rm flow}^{(i)}=O((n+m)^2(n+L))$.
\end{theorem}

\emph{Proof.}
Both algorithms begin with greedy initialization, which
processes each hyperedge by sorting its incident vertices
according to their current indegrees. This requires
$O(\sum_{e\in E}|e|\log(|e|+1))$ time.
In each balancing round, Flow-Decomp computes the
maximum reachable indegrees through SCC condensation
and DAG propagation, while Reverse-Decomp uses direct
reverse traversal. Both procedures take $O(B)$ time.
The same bound applies to the final IDN computation.
Hence, the total cost of reachability computation over
$q$ balancing rounds and the final extraction is
$O((q+1)B)$.
Each balancing round constructs or updates a flow network
with $O(n+m)$ vertices and $O(n+L)$ arcs.
Using Dinic's algorithm, the maximum-flow computation
takes $O((n+m)^2(n+L))$ time per round.
To bound the number of rounds, consider the potential
$\Psi=\sum_{v\in V}\vec d_v^2$.
Every positive balancing round decreases $\Psi$ by at
least two while preserving the total indegree $R_\delta$.
Since $0\le\Psi\le R_\delta^2$, the number of positive
balancing rounds satisfies $q\le R_\delta^2/2$.
Finally, both algorithms store the hypergraph incidences,
vertex indegrees, auxiliary reachability structures,
and the residual flow network, requiring $O(B)$ space.
Combining these bounds establishes the theorem.
\hfill$\square$

\section{Sandwich Bounds and Acceleration}
\label{sec:sandwich}

The algorithms in Section~\ref{sec:two-algorithms} compute the complete density decomposition through global flow balancing. However, some vertices can be assigned to their density levels without completing the balancing process. Based on this observation, we introduce an optional optimization for both \kw{Flow\text{-}Decomp} and \kw{Reverse\text{-}Decomp}. We derive lower and upper reachability bounds that identify vertices whose membership can be determined in advance and restrict the remaining computation to a smaller set of vertices. We further show that the resulting subproblem preserves the original capped hyperedge contribution objective, allowing either decomposition algorithm to operate on the reduced instance without affecting correctness.

\subsection{Lower and Upper Reachability Bounds}
Consider any legal $\delta$-orientation, not necessarily egalitarian, and include the length-zero path in reachability. Define
\begin{equation*}
 \ell(v)=\min_{u:u\rightsquigarrow v}\vec d_u,
 \qquad
 u(v)=\max_{w:v\rightsquigarrow w}\vec d_w.
\end{equation*}
For an integer $k$, let $A_k=\{v:\ell(v)\ge k\}$ and $B_k=\{v:u(v)\ge k\}$. Both sets are predecessor-closed. The lower quantity is obtained by propagating minimum loads forward, while the upper quantity is the reachable maximum already used by the decomposition algorithms.

\begin{lemma}[Contribution--Indegree Equality]
\label{lem:closed-allocation}
If $C$ is predecessor-closed in a valid
$\delta$-orientation, then
$F_\delta(C)=\sum_{v\in C}\vec d_v$.
\end{lemma}
\emph{Proof.}
Consider any hyperedge $e\in E$ with head set $H_e$ and tail set $T_e$. If $H_e\cap C\ne\emptyset$, the predecessor-closed property of $C$ implies $T_e\subseteq C$. Thus, every vertex in $e\setminus C$ belongs to $H_e$, and the number of heads in $C$ is
$|H_e\cap C|=r_e-|e\setminus C|
=\max\{0,r_e-|e\setminus C|\}$.
Otherwise, if $H_e\cap C=\emptyset$, all $r_e$ heads lie outside $C$, giving $|e\setminus C|\ge r_e$. Hence, both the number of heads in $C$ and the capped contribution of $e$ to $C$ are zero.
Therefore, the contribution of each hyperedge to $F_\delta(C)$ equals its number of heads in $C$. Summing over all hyperedge occurrences yields
$F_\delta(C)=\sum_{e\in E}|H_e\cap C|
=\sum_{v\in C}\vec d_v$,
which proves the result.
\hfill$\square$

\begin{theorem}[Sandwich theorem]\label{thm:sandwich}
For every legal $\delta$-orientation and integer $k\ge0$,
$A_k\subseteq D_{k,\delta}\subseteq B_k$.
The inclusions may be strict and hold before global balancing.
\end{theorem}
\emph{Proof.} Write $Q=Q_{k,\delta}$. Every vertex of $A_k$ has load at least $k$. For nonempty $X\subseteq A_k$, Lemma~\ref{lem:closed-allocation} and the fact that $F_\delta(A_k\setminus X)$ is at most the number of heads assigned to $A_k\setminus X$ give
$F_\delta(A_k)-F_\delta(A_k\setminus X)
 \ge\sum_{v\in X}\vec d_v\ge k|X|>(k-1)|X|$.
Thus $Q(A_k)>Q(A_k\cap Z)$ whenever $A_k\setminus Z$ is nonempty. If $Z$ is a maximizer, supermodularity and $Q(A_k\cup Z)\le Q(Z)$ imply $Q(A_k)\le Q(A_k\cap Z)$, a contradiction. Hence $A_k$ is contained in every maximizer and therefore in their intersection $D_{k,\delta}$.
For $Y\subseteq V\setminus B_k$, every $v\in Y$ has $\vec d_v\le k-1$. Lemma~\ref{lem:closed-allocation} and the allocation upper bound yield
$F_\delta(B_k\cup Y)-F_\delta(B_k)
 \le\sum_{v\in Y}\vec d_v\le(k-1)|Y|$.
For any maximizer $Z$, supermodularity together with this inequality shows that $Z\cap B_k$ is also a maximizer. The unique smallest maximizer is contained in it, so $D_{k,\delta}\subseteq B_k$. \hfill$\square$

At an egalitarian orientation, $u(v)$ is the IDN by Theorem~\ref{thm:idn-reach}, and hence $B_k=D_{k,\delta}$ for every $k$. The lower certificate can remain strict because a predecessor path to a seed may contain a vertex of load $k-1$, so in general only $A_k\subseteq D_{k,\delta}=B_k$ is guaranteed. The upper bounds alone therefore finalize all memberships at an egalitarian orientation. Before that point, $B_k\setminus A_k$ contains every vertex whose membership is not certified by the two bounds.

\subsection{Contribution-Preserving Reduction}
Simply inducing the hypergraph on $B_k\setminus A_k$ is incorrect because $F_\delta$ counts partially contained edges. For any predecessor-closed $A\subseteq B$, put $R=B\setminus A$. Each original occurrence $e$ becomes $e_R=e\cap R$ with residual budget
\begin{equation*}
 r_e^{A,B}=\min\{|e_R|,\max\{0,r_e-|e\setminus B|\}\}.
\end{equation*}
Occurrences with zero residual budget are omitted; equal residual occurrences remain distinct.

\begin{theorem}[Contribution Equivalence]
\label{thm:residual-objective}
Let $A\subseteq B$ be predecessor-closed vertex sets and $R=B\setminus A$. Let $F^{A,B}_\delta$ denote the capped hyperedge contribution function of the reduced hypergraph defined above. For any $Y\subseteq R$, we have
$F_\delta(A\cup Y)-F_\delta(A)=F^{A,B}_\delta(Y)$.
Moreover, restricting the current $\delta$-orientation to $R$ preserves the required head count $r_e^{A,B}$ for each retained hyperedge, yielding a valid orientation of the reduced hypergraph.
\end{theorem}
\emph{Proof.}
We first prove the contribution identity. Consider any hyperedge $e\in E$, and let $c=|e\setminus B|$ and $s=|e_R|$, where $e_R=e\cap R$. Since $A\subseteq B$ and $Y\subseteq R=B\setminus A$, we have $|e\setminus(A\cup Y)|=c+|e_R\setminus Y|$,
and $|e\setminus A|=c+s$. Thus, the difference in the contribution of $e$ is
$g_e(A\cup Y)-g_e(A) =\max\{0,r_e-c-|e_R\setminus Y|\} -\max\{0,r_e-c-s\}$,
where $g_e(X)=\max\{0,r_e-|e\setminus X|\}$.
We consider three cases according to the value of $r_e-c$.
If $r_e-c\le 0$, both terms are zero, and $r_e^{A,B}=0$. Hence, the contribution difference is zero.
If $0<r_e-c<s$, the second term vanishes and $r_e^{A,B}=r_e-c$. The difference becomes
$\max\{0,r_e^{A,B}-|e_R\setminus Y|\}$.
If $r_e-c\ge s$, the two terms differ by
$s-|e_R\setminus Y|=|e_R\cap Y|$. Since $r_e^{A,B}=s$, this is also equal to
$\max\{0,r_e^{A,B}-|e_R\setminus Y|\}$.
Therefore, in all cases,
$g_e(A\cup Y)-g_e(A)
=\max\{0,r_e^{A,B}-|e_R\setminus Y|\}$.
Summing over all hyperedge occurrences proves
$F_\delta(A\cup Y)-F_\delta(A)
=F_\delta^{A,B}(Y)$.

We next prove that the restricted orientation remains valid. Since $A$ and $B$ are predecessor-closed, the edgewise equality established in Lemma~\ref{lem:closed-allocation} gives
$|H_e\cap A|=g_e(A),
|H_e\cap B|=g_e(B)$.
As $R=B\setminus A$, the number of heads retained in $R$ is
$|H_e\cap R|=|H_e\cap B|-|H_e\cap A|=\max\{0,r_e-c\} -\max\{0,r_e-c-s\}=r_e^{A,B}$.
Thus, each retained hyperedge has $r_e^{A,B}$ heads after restriction, with all remaining vertices forming its tail set. The resulting orientation satisfies the required head budgets and is therefore valid for the reduced hypergraph.
\hfill$\square$

Taking $A=A_k$ and $B=B_k$, Theorem~\ref{thm:sandwich} restricts the maximization to sets $A_k\cup Y$ with $Y\subseteq B_k\setminus A_k$. The residual identity shows that the smallest maximizing $Y$ is the solution at $k$ on the residual subproblem. Thus the full answer is $D_{k,\delta}=A_k\cup Y$.

\subsection{Sandwich-Accelerated Decomposition}
Algorithm~\ref{alg:sandwich} recursively computes the density decomposition over an integer IDN interval $[a,b]$. At each step, it selects a density level $k$ and computes the lower and upper reachability bounds using either SCC condensation or reverse propagation. These bounds determine the membership of vertices outside the undecided set $B_k\setminus A_k$, allowing the remaining computation to be restricted to a smaller subproblem. The algorithm then applies integral flow to determine $D_{k,\delta}$, using the adjusted hyperedge budgets to preserve the original contribution function. The resulting dense set partitions the current vertices into two disjoint subproblems, corresponding to IDN intervals $[k,b]$ and $[a,k-1]$. Each subproblem is processed recursively with updated hyperedge budgets. If the current orientation is already egalitarian, the remaining IDNs are obtained directly from the upper reachability bounds without further recursion.

\begin{algorithm}[t]
\small
\caption{$\kw{Sandwich}(\mathcal H,\delta,[a,b])$}
\label{alg:sandwich}
\KwData{An oriented hypergraph $\vec{\mathcal H}$,
        edge budgets $r_e$, and IDN interval $[a,b]$}
\KwResult{IDNs $\bar r_v^\delta$ for all $v\in V$}

\If{$a=b$}{
    Assign $\bar r_v^\delta\gets a$ for all $v\in V$\;
    \Return\;
}

Compute $\ell(v)$ and $u(v)$ by SCC condensation
or reverse propagation\;

\If{$u(v)\le\vec d_v+1$ for all $v\in V$}{
    Assign $\bar r_v^\delta\gets u(v)$
    for all $v\in V$\;
    \Return\;
}

$k\gets a+\lfloor(b-a)/2\rfloor+1$\;

$A\gets\{v:\ell(v)\ge k\}$;
$B\gets\{v:u(v)\ge k\}$\;

Construct the reduced hypergraph on $B\setminus A$
with budgets $r_e^{A,B}$\;

Apply integral flow at pivot $k-1$ and recover the residual
dense set $Y_k\subseteq B\setminus A$ by reverse reachability\;

$D_{k,\delta}\gets A\cup Y_k$; lift the residual head assignments
to the current orientation\;

Construct the two contribution-preserving child subproblems induced by
$D_{k,\delta}$ and $V\setminus D_{k,\delta}$\;

Recursively process them with $[k,b]$ and $[a,k-1]$, respectively\;

\Return{the computed IDNs}\;
\end{algorithm}

\begin{theorem}[Correctness and complexity]\label{thm:sandwich-algorithm}
Algorithm~\ref{alg:sandwich}, instantiated with either \kw{Flow\text{-}Decomp} or \kw{Reverse\text{-}Decomp}, returns the canonical IDN array at $\delta$. At one recursion depth, its residual subproblems contain at most $n$ vertices and at most $L$ incidences in total. The depth is at most $\lceil\log_2(K+1)\rceil$, where $K=\max_v d_v$; bound construction therefore costs $O((n+m+L)\log(K+1))$ in total, in addition to the flow calls on the restricted subproblems. Working space is $O(n+m+L)$ when parent storage is released before processing its children.
\end{theorem}
\emph{Proof.} Theorem~\ref{thm:sandwich} determines membership outside the undecided set, and Theorem~\ref{thm:residual-objective} makes the residual problem at $k$ identical to the original objective on all remaining candidates. Induction over the recursion tree proves every split at $k$ and hence every returned IDN. Each split partitions the current vertices. An original incidence belongs to at most one child, so the aggregate size per depth is linear. Midpoint splitting gives the depth bound, and a depth-first worklist permits the stated space bound. \hfill$\square$

\section{Extension to Dynamic Updates}
For a hyperedge $e$, let $g_e(X)=\max\{0,r_e-|e\setminus X|\}$ denote its contribution to a vertex set $X$. Under hyperedge insertion or deletion, the capped contribution function is updated as
\[
F_\delta^+(X)=F_\delta(X)+g_e(X),\qquad
F_\delta^-(X)=F_\delta(X)-g_e(X).
\]
When deleting a hyperedge, only the specified edge is removed, while other hyperedges with identical vertex sets are retained. Newly introduced vertices are treated as isolated vertices in the original hypergraph, and deleting hyperedges does not remove their incident vertices. The updated dense sets are defined as the unique inclusionwise smallest maximizers of their corresponding updated objectives.

\begin{theorem}[IDN Change Bounds]\label{thm:dynamic-capped}
For the insertion of a hyperedge and $k\ge1$,
$D_{k,\delta}\subseteq D^+_{k,\delta}
\subseteq D_{k-1,\delta}$.
For deletion,
$D_{k+1,\delta}\subseteq D^-_{k,\delta}
\subseteq D_{k,\delta}$.
Consequently, inserting a hyperedge increases each vertex's IDN by at most one, while deleting a hyperedge decreases it by at most one.
\end{theorem}

\emph{Proof.}
We first establish a general property of supermodular functions. Let $Q$ and $Q+h$ be supermodular, where $h$ is nondecreasing, and let $A$ and $B$ be their respective smallest maximizers. Suppose $A\not\subseteq B$. Since $A\cap B$ is a proper subset of $A$, we have $Q(A)>Q(A\cap B)$. By supermodularity,
$Q(A\cup B)-Q(B)
\ge Q(A)-Q(A\cap B)>0$.
As $h$ is nondecreasing, this implies
$(Q+h)(A\cup B)>(Q+h)(B)$, contradicting the optimality of $B$. Hence, $A\subseteq B$.
For hyperedge insertion, both $g_e(X)$ and $|X|-g_e(X)$ are nondecreasing, since adding a vertex increases $g_e$ by either zero or one. Applying the above property to $Q_{k,\delta}$ and $Q^+_{k,\delta}=Q_{k,\delta}+g_e$ gives
$D_{k,\delta}\subseteq D^+_{k,\delta}$.
Moreover,
$Q^+_{k,\delta}(X)+|X|-g_e(X)
=Q_{k-1,\delta}(X)$,
which yields $D^+_{k,\delta}\subseteq D_{k-1,\delta}$.
All these objectives are supermodular by the supermodularity of their corresponding capped contribution functions.
For deletion, applying the insertion result to the updated hypergraph gives
$D^-_{k,\delta}\subseteq D_{k,\delta}$ and
$D_{k+1,\delta}\subseteq D^-_{k,\delta}$.
The stated IDN bounds follow directly from these inclusions and the nesting of dense sets.
\hfill$\square$

The one-unit bounds concern IDNs, not layer densities. Updated membership can change both sets in the layer-density numerator; old layer averages cannot simply be incremented or decremented at edited heads. Layer statistics, when requested, are reevaluated using the updated function and IDNs by Proposition~\ref{prop:layer-evaluation}.

Updating vertex indegrees alone is insufficient to maintain the IDNs. In particular, deleting a hyperedge may change the IDN of a vertex in its tail set even when its indegree remains unchanged. We therefore retain the current $\delta$-orientation and eliminate reversible hyperpaths within the affected weakly connected components.
For an updated hyperedge $e$, let $A=e$ denote its vertex set, recorded before deletion if necessary. After the update, let $C(A)$ be the union of weakly connected components in the updated undirected incidence graph that contain at least one vertex of $A$. If deleting an edge disconnects a component, each affected component still contains a vertex of the deleted edge. For batch updates, $A$ includes the vertices of all modified hyperedges.
Algorithms~\ref{alg:insert} and~\ref{alg:delete} retain the existing orientation and perform flow balancing only within $C(A)$. Once no reversible hyperpath remains, the updated IDNs are computed using SCC condensation or reverse propagation. Components outside $C(A)$ remain unchanged and require no further balancing.

\begin{algorithm}[t]
\small
\caption{$\kw{INSERT}(\vec{\mathcal H},\bar r^\delta,e)$}
\label{alg:insert}
\KwData{An egalitarian $\delta$-orientation
$\vec{\mathcal H}$, vertex IDNs $\bar r^\delta$,
and a new hyperedge $e$}
\KwResult{Updated orientation and IDNs}

Initialize newly added vertices with zero indegree
and IDN\;

Order $e=(u_1,u_2,\ldots,u_{|e|})$ such that
$\vec d_{u_i}\le\vec d_{u_{i+1}}$\;

$r_e\gets\min\{\delta,|e|\}$,
$H_e\gets\{u_1,\ldots,u_{r_e}\}$\;

Insert $\vec e=(e\setminus H_e,H_e)$
into $\vec{\mathcal H}$\;

\ForEach{$v\in H_e$}{
    $\vec d_v\gets\vec d_v+1$\;
}

$C\gets$ all connected components containing
vertices of $e$\;

\While{a reversible hyperpath
$s\rightsquigarrow t$ exists in $C$}{
    Choose $q$ with $\vec d_s<q<\vec d_t$\;

    Construct or reset the flow network at $q$\;

    Compute an integral maximum flow and
    update the orientation and indegrees\;
}

Compute $\bar r_v^\delta$ for all $v\in V$
by SCC condensation or reverse propagation\;

\Return{$(\vec{\mathcal H},\bar r^\delta)$}\;
\end{algorithm}

\begin{algorithm}[t]
\small
\caption{$\kw{DELETE}(\vec{\mathcal H},\bar r^\delta,e)$}
\label{alg:delete}
\KwData{An egalitarian $\delta$-orientation
$\vec{\mathcal H}$, vertex IDNs $\bar r^\delta$,
and a hyperedge $e$ to be deleted}
\KwResult{Updated orientation and IDNs}

Let $\vec e=(T_e,H_e)$ and $A\gets e$\;

Remove $\vec e$ from $\vec{\mathcal H}$\;

\ForEach{$v\in H_e$}{
    $\vec d_v\gets\vec d_v-1$\;
}

$C\gets$ all connected components containing
vertices of $A$ in the updated hypergraph\;

\While{a reversible hyperpath
$s\rightsquigarrow t$ exists in $C$}{
    Choose $q$ with $\vec d_s<q<\vec d_t$\;

    Construct or reset the flow network at $q$\;

    Compute an integral maximum flow and
    update the orientation and indegrees\;
}

Compute $\bar r_v^\delta$ for all $v\in V$
by SCC condensation or reverse propagation\;

\Return{$(\vec{\mathcal H},\bar r^\delta)$}\;
\end{algorithm}

\begin{theorem}[Correct dynamic maintenance]\label{thm:dynamic}
For any $\delta$, Algorithms~\ref{alg:insert} and~\ref{alg:delete} terminate and recover an egalitarian orientation and correct IDNs. Repairs need only start inside $C(A)$. This is a correctness and locality-of-search guarantee, \emph{not} a worst-case sublinear update-time guarantee.
\end{theorem}
\emph{Proof.}
After each update, every hyperedge has the required number of
heads, so the resulting orientation remains a valid
$\delta$-orientation. Since the updated contribution function
is additive across weakly connected components, its smallest
maximizer can be obtained from the smallest maximizers of
individual components. Outside $C(A)$, neither the hypergraph
nor its orientation changes. These components therefore remain
egalitarian, and their IDNs are unchanged.
Any new reversible hyperpath must lie within $C(A)$.
Suppose such a path runs from $s$ to $t$, with
$a=\vec d_s$ and $b=\vec d_t$, where $b-a\ge2$.
We choose an integer $k$ satisfying $a+2\le k\le b$
and set $q=k-1$. The flow network at pivot $q$
admits a positive flow that adjusts head assignments
from vertices above $q$ to those below $q$, without
changing the head budget of any hyperedge.
Each unit of flow decreases the potential
$\Psi=\sum_{v\in V}\vec d_v^2$ by at least two.
Since $\Psi$ is a nonnegative integer, the balancing
process terminates after finitely many steps.
At termination, no reversible hyperpath remains in
$C(A)$. As the other components are already egalitarian,
the updated orientation is egalitarian over the entire
hypergraph. By Theorem~\ref{thm:idn-reach}, the final
reachability computation yields the correct IDNs,
hence the canonical dense sets of the updated hypergraph.
\hfill$\square$

The residual network can be reused across balancing rounds for the same hypergraph, but must be rebuilt after an insertion or deletion, since the factor vertices, incidence structure, and assigned head contributions may change. For each update, the flow network is constructed only when needed and reused in subsequent rounds. The potential $\Psi$ decreases during balancing, but may increase or decrease after the next hyperedge update.

For successive updates, we retain the current orientation across batches to avoid repeated initialization. A chosen positive $\delta$ remains valid when $d^E_{\max}$ changes. If the parameter range expands, newly introduced values of $\delta$ require separate orientations. Identical hyperedges are treated as distinct records and deleted using stable IDs, while repeated vertices within a hyperedge are removed during normalization. In the worst case, an update may affect the entire hypergraph. Section~\ref{sec:rerun-dynamic} compares component-based maintenance with recomputation on the same updated graphs.

\begin{theorem}[Complexity of Dynamic Maintenance]
\label{thm:dynamic-complexity}
Let $M$ denote the total number of hyperedge slots allocated
and $L_a$ the number of active incidences in the updated
hypergraph. For a normalized hyperedge insertion or deletion,
let $q$ be the number of flow-balancing rounds and
$T_{\rm flow}^{(i)}$ the maximum-flow computation time
in round $i$. Algorithms~\ref{alg:insert} and~\ref{alg:delete}
require
\[
O\left(
(q+1)(n+M+L_a)
+\sum_{i=1}^{q}T_{\rm flow}^{(i)}
\right)
\]
time and $O(n+M+L_a)$ space.
\end{theorem}

\emph{Proof.}
Both algorithms retain the current $\delta$-orientation and
update only the hyperedges affected by an insertion or deletion.
Identifying the affected connected components requires
$O(n+M+L_a)$ time in the worst case.
Each balancing round computes reachable indegrees using
SCC condensation or reverse propagation, taking
$O(n+M+L_a)$ time. The corresponding maximum-flow
computation takes $T_{\rm flow}^{(i)}$ time.
After $q$ rounds, the final IDNs are computed in
$O(n+M+L_a)$ time. Summing these costs gives the
stated time bound.
The maintained hypergraph, vertex indegrees, incidence
structures, and flow network require $O(n+M+L_a)$ space.
Deleted hyperedges retain their identifiers but release
their incidence storage. Thus, the space bound depends
on the number of allocated hyperedge slots and active
incidences rather than all historical incidences.
\hfill$\square$

%For an unnormalized hyperedge of size $s_e$, sorting and deduplication require an additional $O(s_e\log(s_e+1))$ time. Since the affected components may span the entire hypergraph, the algorithms do not guarantee sublinear update time. Moreover, reporting all updated IDNs requires $\Omega(n)$ time. Over a long update sequence, retaining stable hyperedge identifiers may further increase the cost of traversing the maintained hypergraph.

\section{Experimental Evaluation}
\label{sec:reproduced}
\subsection{Experimental setup}

We evaluate the proposed algorithms on multiple datasets. All programs are written in C++ and compiled with GCC 11.1.0 using \texttt{-O3}. Experiments are conducted on a Linux server with a 2.9\,GHz AMD Ryzen 3990X CPU and 256\,GB RAM running CentOS 7.9.2 (64-bit).

%In this section, we evaluate the performance of our algorithms across multiple datasets.
%All programs are in C++ (gcc 11.1.0, -O3) and executed on a Linux server with a 2.9 GHz AMD Ryzen 3990X\kw{CP}U and 256 GB RAM (CentOS 7.9.2, 64-bit).

\noindent{\textbf{Datasets}}.
We evaluate our methods on nine benchmark hypergraphs, whose normalized statistics are summarized in Table~\ref{tab:new-data}.
The first seven datasets are available from the Cornell Hypergraph Repository at \url{https://www.cs.cornell.edu/~arb/data/}.
\kw{HC} (\textit{house-committees}) represents memberships of US legislators in congressional committees.
\kw{CP} (\textit{contact-primary-school}) and \kw{CH} (\textit{contact-high-school}) record groups of students observed in face-to-face proximity.
\kw{SB} (\textit{senate-bills}) and \kw{HB} (\textit{house-bills}) group legislators who co-sponsor the same bill.
\kw{TC} (\textit{trivago-clicks}) groups hotels clicked within the same browsing session, whereas \kw{SA} (\textit{stackoverflow-answers}) groups questions answered by the same user.
The remaining two datasets use pinned XGI HIF releases~\cite{xgiDblp2026,xgiMath2026}.
\kw{TM} (\textit{threads-math-sx}) groups users participating in Math Stack Exchange threads, and \kw{DBLP} groups the authors of each publication.

\begin{table}[t]\centering\small
\caption{Data hypergraphs.}\label{tab:new-data}
\setlength{\tabcolsep}{3.5pt}\begin{tabular}{lrrrr}\toprule
Data & $n$ & $m$ & $L$ & $\max |e|$\\\midrule
\kw{HC} & 1,290 & 341 & 11,843 & 81\\
\kw{CP} & 242 & 12,704 & 30,729 & 5\\
\kw{CH} & 327 & 7,818 & 18,192 & 5\\
\kw{SB}& 294 & 29,157 & 232,147 & 99\\
\kw{HB}& 1,494 & 60,987 & 1,248,666 & 399\\
\kw{TC} & 172,738 & 233,202 & 726,861 & 85\\
\kw{SA} & 15,211,989 & 1,103,243 & 26,109,177 & 61,315\\
\kw{TM} & 176,445 & 719,792 & 1,610,393 & 21\\
\kw{DBLP} & 1,930,378 & 3,700,681 & 10,326,773 & 280\\
\bottomrule\end{tabular}\end{table}

\stitle{Algorithms.} We evaluate three $(k,\delta)$-dense subhypergraph search algorithms: $\kw{DSM\text{-}PATH}$, $\kw{DSM\text{-}FLOW}$, and $\kw{DSM\text{-}FLOW+}$ (Algorithms~\ref{dshs}--\ref{dshs3}). The two complete decomposition algorithms are $\kw{Flow\text{-}Decomp}$ and $\kw{Reverse\text{-}Decomp}$ (Algorithms~\ref{alg:scc-all} and~\ref{alg:rr-reach}). The Sandwich procedure in Algorithm~\ref{alg:sandwich} supplies certified bounds to either decomposition algorithm; the resulting executions are denoted \kw{Flow\text{-}S} and \kw{Reverse\text{-}S}. Algorithms~\ref{alg:insert} and~\ref{alg:delete} maintain the decomposition after hyperedge insertions and deletions. We compare the resulting hierarchy with degree core~\cite{9458645}, nbr-$k$-core~\cite{DBLP:journals/pvldb/ArafatKRG23}, $(k,h)$-core~\cite{DBLP:conf/icde/LuoYLZ0023}, and $(\alpha,\beta)$-core~\cite{abcore}. We use $h=5$, $\beta=5$, and $\delta=5$ in the cross-model experiments. 

\begin{table}[t]\centering\footnotesize
\caption{Subhypergraph search time.}\label{tab:search-three}
\setlength{\tabcolsep}{3.5pt}\begin{tabular}{lrrr}\toprule
Data & \kw{DSM\text{-}PATH} & \kw{DSM\text{-}FLOW} & \kw{DSM\text{-}FLOW+}\\\midrule
HC & 2.073 & 0.005 & 0.005\\
\kw{CP} & 0.009 & 0.009 & 0.010\\
CH & 0.006 & 0.006 & 0.007\\
\kw{SB}& 3776.991 & 0.040 & 0.042\\
\kw{HB}& TO & 0.147 & 0.159\\
\kw{TC} & 7952.924 & 0.741 & 0.660\\
\kw{SA} & TO & 25.972 & 25.414\\
\kw{TM} & 7997.339 & 0.597 & 0.594\\
\kw{DBLP} & TO & 6.581 & 6.500\\
\bottomrule\end{tabular}\end{table}

\subsection{Efficiency Evaluation}

\stitle{Exp-1: Runtime of dense subhypergraph search algorithms.}
Table~\ref{tab:search-three} compares \kw{DSM\text{-}PATH}, \kw{DSM\text{-}FLOW}, and \kw{DSM\text{-}FLOW+} (Algorithms~\ref{dshs}--\ref{dshs3}) on nine hypergraphs with $k=5$ and $\delta=5$. Here, TO indicates that an algorithm fails to return a result within 12 hours. The results show that both flow-based algorithms substantially outperform \kw{DSM\text{-}PATH}, particularly on large hypergraphs. On SB, \kw{DSM\text{-}PATH} takes 3776.991 seconds, whereas \kw{DSM\text{-}FLOW} completes in only 0.040 seconds. Similar improvements are observed on \kw{TC} and \kw{TM}, where \kw{DSM\text{-}FLOW} finishes within one second, compared with over 7900 seconds for \kw{DSM\text{-}PATH}. On HB, \kw{SA}, and DBLP, \kw{DSM\text{-}PATH} times out, while both flow-based algorithms complete successfully. These improvements stem from their different balancing strategies. \kw{DSM\text{-}PATH} repeatedly searches for and reverses individual hyperpaths to obtain a globally egalitarian orientation, resulting in substantial repeated computation. In contrast, \kw{DSM\text{-}FLOW} exploits the weaker balancing condition in Lemma~\ref{lem:kstable} and uses integral maximum flow to adjust multiple head assignments simultaneously, avoiding repeated hyperpath searches and unnecessary global balancing. \kw{DSM\text{-}FLOW+} further improves the initial orientation through greedy head assignment, reducing the required flow adjustment in some cases. Figure~\ref{fig:search-scale} further examines the scalability of the three algorithms.

\stitle{Exp-2: Scalability of subhypergraph mining algorithms.}
Figure~\ref{fig:search-scale} reports the scalability of \kw{DSM\text{-}PATH}, \kw{DSM\text{-}FLOW}, and \kw{DSM\text{-}FLOW+} on \kw{HC} by varying the numbers of vertices and hyperedges. Similar trends are observed on the other datasets. Both flow-based algorithms consistently outperform \kw{DSM\text{-}PATH}, with the performance gap becoming more pronounced as the hypergraph grows. At full scale, \kw{DSM\text{-}PATH} takes 2.217 and 2.008 seconds in the vertex and hyperedge experiments, respectively, whereas both flow-based algorithms complete in less than 0.01 seconds. Their running times remain relatively stable as the hypergraph size increases, indicating better scalability.
The performance difference mainly arises from their reorientation strategies. %As the hypergraph grows, \kw{DSM\text{-}PATH} may require more reversible hyperpath operations, each involving reachability analysis and path reconstruction. These repeated searches incur substantial computational overhead. In contrast, \kw{DSM\text{-}FLOW} and \kw{DSM\text{-}FLOW+} exploit the weaker balancing condition in Lemma~\ref{lem:kstable} and use maximum flow to adjust multiple head assignments simultaneously. By avoiding repeated hyperpath searches and unnecessary global balancing, both algorithms maintain low running times as the hypergraph size increases.

\begin{figure}[t]
\centering\includegraphics{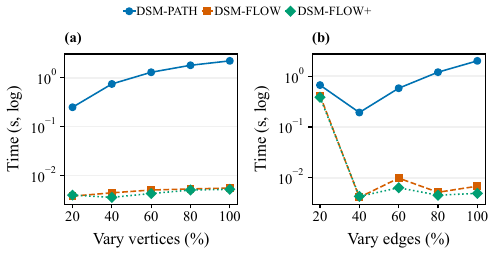}
\caption{Search scalability on \kw{HC}.}\label{fig:search-scale}
\end{figure}

\stitle{Exp-3: Runtime of decomposition algorithms.}
Figure~\ref{fig:new-runtime} compares \kw{Flow\text{-}Decomp} and \kw{Reverse\text{-}Decomp}, together with their Sandwich-accelerated variants (\kw{Flow\text{-}S} and \kw{Reverse\text{-}S}), on nine hypergraphs. \kw{Reverse\text{-}Decomp} consistently outperforms \kw{Flow\text{-}Decomp} on the five model-comparison datasets, achieving speedups of 3.11$\times$ on \kw{SB} and 2.39$\times$ on \kw{HB}. This advantage mainly comes from their different reachability computations. \kw{Flow\text{-}Decomp} repeatedly constructs SCCs and propagates maximum indegrees over the condensation DAG, whereas \kw{Reverse\text{-}Decomp} uses descending-load reverse propagation to compute the same results without SCC construction. This reduces the overhead of repeated reachability computations during flow balancing.
Sandwich acceleration further improves efficiency by identifying vertices whose density levels can be determined in advance and restricting the remaining computation to equivalent reduced hypergraphs. On \kw{SB}, \kw{Flow\text{-}S} and \kw{Reverse\text{-}S} take 0.274 and 0.262 seconds, respectively, while on \kw{HB} they take 1.270 and 1.280 seconds. These results show that reachability-based reduction can reduce the cost of flow balancing. Both accelerated variants return the same IDNs as their original algorithms.

\begin{figure}[t]
\centering\includegraphics{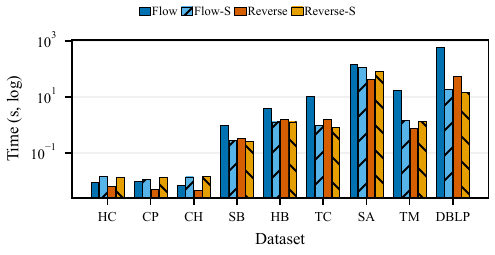}
\caption{Decomposition time on nine datasets.}\label{fig:new-runtime}
\end{figure}

\begin{figure}[t]
\centering\includegraphics{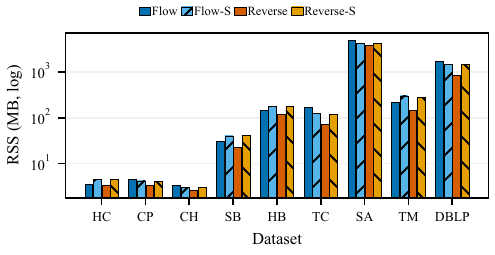}
\caption{Peak memory on nine datasets.}\label{fig:new-memory}
\end{figure}

\stitle{Exp-4: Memory overhead of decomposition algorithms.}
Figure~\ref{fig:new-memory} compares the peak memory consumption of the four decomposition methods on nine hypergraphs. \kw{Reverse\text{-}Decomp} generally requires less memory than \kw{Flow\text{-}Decomp}, particularly on larger datasets. On HB, its peak memory decreases from 147.56 MB to 117.34 MB, a reduction of 20.5\%. This improvement mainly comes from avoiding explicit SCC construction and condensation DAG storage. Instead, \kw{Reverse\text{-}Decomp} uses a compact reverse incidence index and directly propagates reachable indegrees, reducing the auxiliary data structures required during balancing. The Sandwich-accelerated variants maintain comparable memory consumption while reducing the computation through reachability bounds and smaller hypergraph instances. These results show that the proposed optimizations improve computational efficiency without substantially increasing memory requirements.

\stitle{Exp-5: Scalability with vertices and hyperedges.}
Figure~\ref{fig:new-scale} evaluates the scalability of \kw{Flow\text{-}Decomp}, \kw{Reverse\text{-}Decomp}, and their Sandwich-accelerated variants on \kw{HB} by varying the numbers of vertices and hyperedges. Similar trends are observed on other datasets. As the hypergraph grows, \kw{Reverse\text{-}Decomp} and the accelerated variants generally achieve lower running times than \kw{Flow\text{-}Decomp}. At full scale, \kw{Flow\text{-}Decomp} takes 3.998 and 4.322 seconds in the vertex and hyperedge experiments, respectively, compared with 1.639 and 1.770 seconds for \kw{Reverse\text{-}Decomp}. \kw{Flow\text{-}S} further reduces these times to 1.319 and 1.198 seconds, while \kw{Reverse\text{-}S} takes 1.550 and 1.142 seconds.
These improvements mainly result from direct reverse propagation, which avoids repeated SCC construction, and Sandwich acceleration, which reduces the vertices and hyperedges involved in flow balancing. Both techniques reduce computation as the hypergraph grows, improving the scalability of density decomposition.

\begin{figure}[t]
\centering\includegraphics{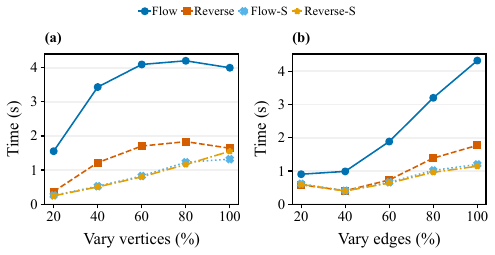}
\caption{Decomposition scalability on \kw{HB}.}\label{fig:new-scale}
\end{figure}

\stitle{Exp-6: Runtime comparison on large datasets.}
Figure~\ref{fig:large-methods} compares four established decomposition methods with our four algorithms on TC, SA, TM, and DBLP. The results show that our optimized algorithms can efficiently compute complete density decompositions on large hypergraphs. The fastest of our algorithms takes 0.801, 42.843, 0.755, and 14.209 seconds on these datasets, respectively. In particular, on SA, which contains over 15 million vertices, the complete decomposition finishes in 42.843 seconds. On \kw{DBLP}, \kw{Reverse\text{-}S} reduces the running time from 591.735 seconds for \kw{Flow\text{-}Decomp} to 14.209 seconds, achieving a speedup of 41.6$\times$.
These improvements mainly come from direct reverse propagation and Sandwich acceleration. \kw{Reverse\text{-}Decomp} avoids repeated SCC construction, while Sandwich acceleration identifies vertices whose density levels can be determined in advance and restricts flow balancing to smaller hypergraphs. Their combination substantially reduces the cost of complete decomposition on large datasets. Although the established methods have lower running times, they construct different core-based hierarchies rather than the proposed density decomposition based on capped hyperedge contributions. The results demonstrate that our optimizations make density-oriented hypergraph decomposition practical on large datasets.

\begin{figure}[t]
\centering\includegraphics{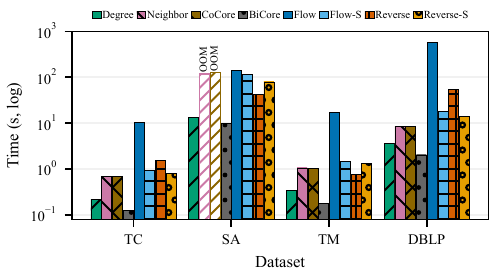}
\caption{Decomposition time on four large datasets.}\label{fig:large-methods}
\end{figure}

\stitle{Exp-7: Memory overhead on large datasets.}
Figure~\ref{fig:large-methods-memory} compares the peak memory consumption of four established decomposition methods and our four algorithms on \kw{TC}, \kw{SA}, \kw{TM}, and \kw{DBLP}. \kw{Reverse\text{-}Decomp} generally requires less memory than \kw{Flow\text{-}Decomp}, particularly on \kw{TC} and \kw{DBLP}. This advantage comes from direct reverse propagation, which avoids explicit SCC construction and condensation DAG storage. The Sandwich-accelerated variants also maintain reasonable memory consumption on large hypergraphs. Moreover, all our algorithms successfully complete on \kw{SA}, whereas Neighbor core and CoCore encounter memory exhaustion. These results demonstrate the memory efficiency and scalability of our decomposition algorithms. 

\begin{figure}[t]
\centering\includegraphics{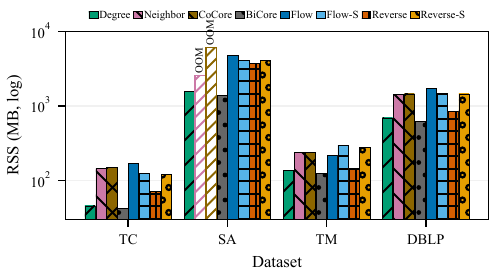}
\caption{Peak memory on four large datasets.}\label{fig:large-methods-memory}
\end{figure}

\begin{figure}[t]
\centering\includegraphics{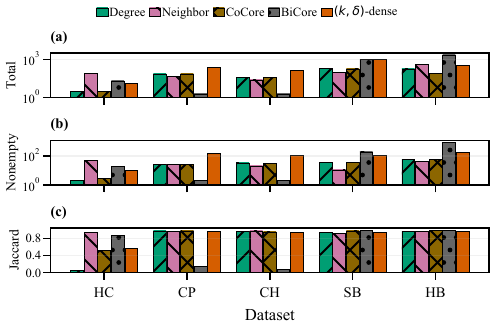}
\caption{Hierarchy statistics across decomposition models.}\label{fig:hierarchy}
\end{figure}

\subsection{Effectiveness Evaluation}
\label{sec:model-comparison}
\stitle{Exp-8: Hierarchy structure across models.}
Figure~\ref{fig:hierarchy} compares the hierarchical structures produced by different decomposition models in terms of total indexed levels, non-empty levels, and adjacent-layer consistency. Total indexed levels represent the range of density scores, while non-empty levels measure the number of distinct layers. Adjacent-layer consistency is measured by the average Jaccard similarity between consecutive non-empty nested sets, reflecting their structural overlap.
The proposed model produces more non-empty levels than degree core on all five datasets, indicating a finer distinction of density levels. This is because degree core relies on local incidence counts, whereas our model considers capped hyperedge contributions and balanced indegrees to distinguish density variations arising from higher-order interactions. Other models may produce more levels under different structural constraints. For example, neighborhood core has 48 levels on \kw{HC} compared with our 10, while bipartite core has 862 levels on \kw{HB} compared with our 175. These results illustrate the differences in hierarchical structure captured by the decomposition models.

\begin{figure}[t]
\centering\includegraphics{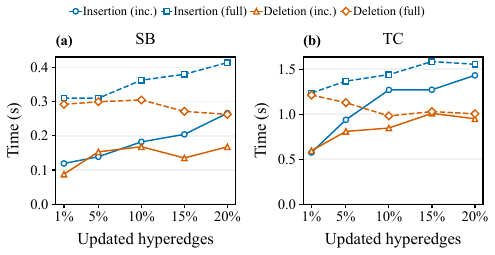}
\caption{Dynamic maintenance on \kw{SB} and \kw{TC}.}\label{fig:dynamic-selected}
\end{figure}

\subsection{Dynamic Updates}
\label{sec:rerun-dynamic}
\stitle{Exp-9: Dynamic maintenance.}
Figure~\ref{fig:dynamic-selected} compares incremental maintenance with full recomputation under hyperedge insertions and deletions on \kw{SB} and \kw{TC}. The experiments consider batch updates involving 1\%, 5\%, 10\%, 15\%, and 20\% of the original hyperedges. Incremental maintenance consistently outperforms full recomputation, achieving speedups ranging from 1.02$\times$ to 3.32$\times$ across all tested settings. The advantage comes from retaining the current egalitarian orientation and avoiding full initialization; component-local balancing can further reduce work when updates do not span the whole hypergraph. Incremental maintenance produces identical IDNs to full recomputation across all tested updates, confirming that our dynamic algorithms improve efficiency while maintaining the correctness of the density decomposition.

\subsection{Case Studies}

\stitle{Exp-10: Fraud screening in multi-party transactions.}
We recover the AMLSim account labels, transaction stream, and alert groups used in the original case study~\cite{DBLP:conf/nips/AltmanBNEAA23}. Each alert is a hyperedge over its participating accounts. The proposed score is the current \kw{Reverse\text{-}Decomp} IDN at $\delta=5$; degree core and neighborhood core provide comparison scores. The released evaluator combines the score percentile with local transaction activity and applies the same Top-50 review budget to all three models. Each model is generated and evaluated once. Table~\ref{tab:case-fraud} reports Top-50 classification metrics together with full-ranking ROC-AUC and average precision (AP). The proposed hierarchy places 28 fraudulent accounts in the first 50 positions, compared with 10 for degree core and 13 for neighborhood core. Its Top-50 precision is therefore 0.56, and its full-ranking AP is 0.6411. The three ROC-AUC values lie between 0.9391 and 0.9412. This case supports concentrated early retrieval and a higher precision--recall ranking under the released evaluator; it does not show a uniformly higher ROC-AUC.

\begin{table}[t]
\centering\small
\caption{Top-50 fraud screening and full-ranking metrics on AMLSim.}\label{tab:case-fraud}
\setlength{\tabcolsep}{2.0pt}
\begin{tabular}{lccccccc}\toprule
Model & TP & Precision & Recall & F1 & ROC-AUC & AP\\\midrule
$(k,\delta)$-Dense  & 28 & 0.560 & 0.0166 & 0.0323 & 0.9410 & 0.6411\\
Degree core & 10 & 0.200 & 0.0059 & 0.0115 & 0.9391 & 0.6088\\
Neighbor core & 13 & 0.260 & 0.0077 & 0.0150 & 0.9412 & 0.6182\\
\bottomrule\end{tabular}
\end{table}

\begin{figure}[t]
\centering
\subfigure[$(k,10)$-dense]{\includegraphics[width=0.46\columnwidth]{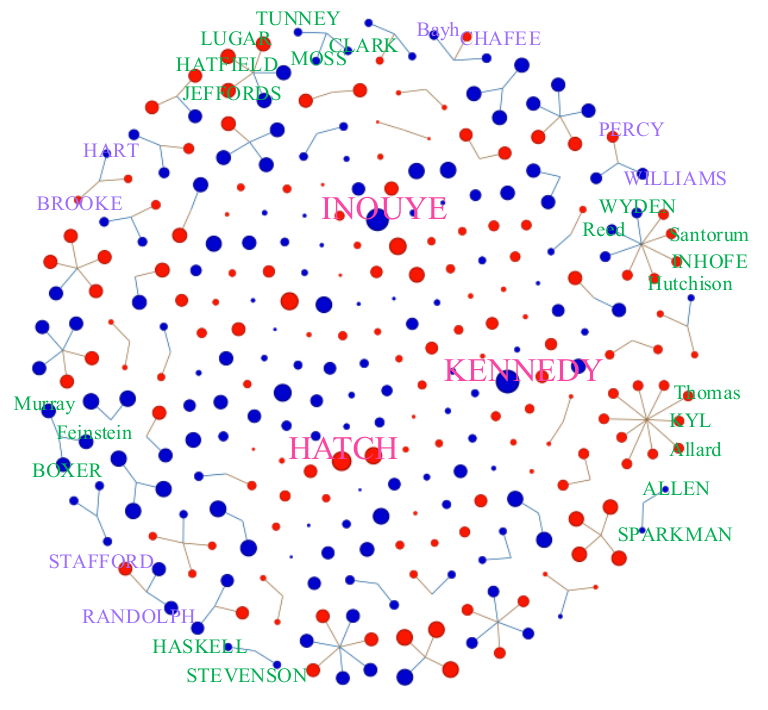}}\hfill
\subfigure[nbr-$k$-core]{\includegraphics[width=0.46\columnwidth]{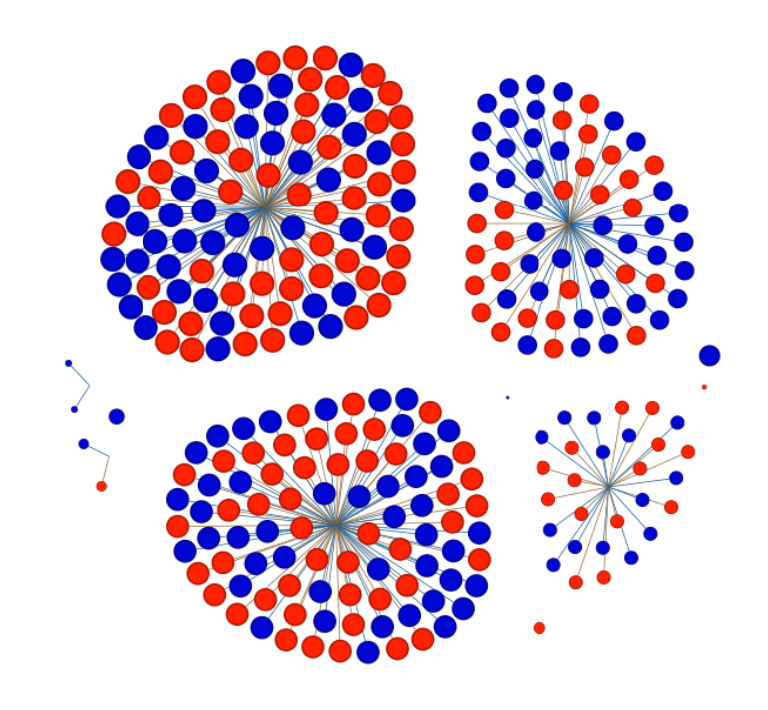}}
\caption{Co-sponsorship community visualization.}\label{fig:case-senate}
\end{figure}

\stitle{Exp-11: Legislative co-sponsorship community discovery.}
We evaluate our method on the U.S.\ Senate Bill dataset, where each bill forms a hyperedge connecting its sponsors. Legislators are represented as vertices and labeled by party affiliation (red for Republicans and blue for Democrats). This dataset captures real political collaborations and allows us to examine the community structures identified by different decomposition models.
Figure~\ref{fig:case-senate} compares our $(k,10)$-dense decomposition ($\delta=10=\lfloor\bar d_e\rfloor$) with nbr-$k$-core. Our model identifies fine-grained communities that distinguish the two political parties while revealing cross-party collaborations. In particular, it highlights influential bipartisan legislators such as \textit{Hatch}, \textit{Kennedy}, and \textit{Inouye}, as well as closely connected within-party groups involving \textit{Murray}, \textit{Feinstein}, and \textit{Boxer}. In contrast, nbr-$k$-core merges these structures into coarser communities, making it difficult to distinguish different collaboration patterns. These results show that our density-based decomposition captures finer hierarchical structures and better reveals political collaboration patterns in legislative hypergraphs.

\section{Related Work}

\stitle{Graph Decomposition.}
Graph decomposition has been extensively studied for identifying cohesive and dense structures. Representative cohesive models include $k$-core~\cite{DBLP:journals/corr/cs-DS-0310049,DBLP:journals/tkde/LiYM14,DBLP:journals/tpds/MontresorPM13}, $k$-truss~\cite{cohen2008trusses,DBLP:journals/pvldb/WangC12,DBLP:conf/sigmod/HuangCQTY14}, nucleus decomposition~\cite{DBLP:conf/www/SariyuceSPC15,DBLP:journals/pvldb/ShiDS21}, and $k$-edge-connected components~\cite{DBLP:conf/sigmod/ChangYQLLL13,DBLP:journals/vldb/ChangW24}. Extensions incorporating distance and attribute constraints include distance-generalized cores~\cite{DBLP:conf/sigmod/BonchiKS19,DBLP:conf/cikm/DaiLQWYZ021} and colorful $h$-star $k$-core decomposition~\cite{DBLP:conf/icde/GaoLQCYW22,DBLP:journals/pvldb/GaoQLH23}.
Density-oriented approaches include densest subgraph discovery~\cite{goldberg1984finding,charikar2000greedy,bahmani2012densest}, locally densest subgraphs~\cite{DBLP:conf/kdd/QinLCZ15,DBLP:journals/tkdd/Tatti19}, higher-order dense subgraph models~\cite{tsourakakis2015k}, density-friendly decomposition~\cite{DBLP:journals/pvldb/ZhangLZQW24}, and bipartite density decomposition~\cite{DBLP:journals/pacmmod/ZhangLZQQW25}. Recent studies further investigate density search over relational graphs, integral density objectives on directed graphs, and locally dense $h$-clique structures~\cite{DBLP:conf/www/Niu0JL25,DBLP:journals/pacmmod/ZhangLLZQW25,DBLP:journals/pvldb/ZhouHF26}. Unlike these methods, which focus on pairwise relations or graph motifs, our work considers hyperedges directly and constructs a density hierarchy based on their capped contributions.

\stitle{Cohesive and Density-Based Hypergraph Decomposition.}
Hypergraph decomposition extends cohesive structure analysis from pairwise relations to higher-order interactions. Representative core-based models include hypergraph $k$-core~\cite{DBLP:conf/ipps/RamadanTP04,bianconi2024nature}, nbr-$k$-core~\cite{DBLP:journals/pvldb/ArafatKRG23,DBLP:journals/pacmmod/ZhangYWLZL25}, $(k,h)$-core~\cite{DBLP:conf/icde/LuoYLZ0023}, $(k,g)$-core~\cite{DBLP:conf/cikm/KimKLJ23}, $(k,q)$-core~\cite{DBLP:conf/cikm/Luo0000Y024,lee2023k}, and $(k,t)$-core~\cite{DBLP:journals/datamine/BuLS23}. These models characterize cohesiveness through degree, neighborhood, co-occurrence, hyperedge-size, and participation constraints. Other approaches, including densest subhypergraph discovery~\cite{DBLP:conf/cikm/HuWC17,huang2024densest} and hypergraph truss decomposition~\cite{DBLP:journals/pvldb/QinZLLYW25}, study density and higher-order structural cohesion.
Recent studies have further explored constrained-core indexing, dynamic core maintenance, community search, locality-based indexing, decay-aware cores, refined $(k,g)$-core decomposition, and overlap-bounded decomposition~\cite{DBLP:journals/vldb/LuoZYYLW25,DBLP:journals/www/LiuLZCY25,DBLP:conf/edbt/LiuLZ00Y25,DBLP:conf/edbt/KimKHKJK26,DBLP:conf/cikm/KimKKK25,DBLP:journals/kbs/KimKKKKLL26,DBLP:journals/vldb/LuoYZZYCLG26}. Unlike core-based methods that primarily impose structural constraints, our model defines a nested density hierarchy through capped hyperedge contributions and balanced orientations, allowing density variations to be characterized directly.

\stitle{Structural hypergraph decomposition.}
Structural hypergraph decompositions focus on acyclicity, constraint satisfaction, and query tractability.
Representative frameworks include $k$-hinge tree decomposition~\cite{jeavons1994structural}, hypertree decomposition~\cite{DBLP:journals/jcss/GottlobLS02}, generalized hypertree decomposition~\cite{DBLP:journals/jcss/CohenJG08}, fractional hypertree decomposition~\cite{DBLP:journals/corr/abs-1711-04506,gottlob2021complexity}, and submodular-width based decomposition~\cite{marx2013tractable}.
These methods target structural tractability, whereas our work focuses on density-aware hierarchical mining in hypergraphs.

\section{Conclusion}
In this paper, we studied density decomposition in hypergraphs and proposed the $(k,\delta)$-dense subhypergraph model based on capped hyperedge contributions. We established the uniqueness and nesting properties of the resulting dense sets and derived an equivalent characterization using egalitarian $\delta$-orientations and reachability. Based on these properties, we developed three dense subhypergraph search algorithms, \kw{DSM\text{-}PATH}, \kw{DSM\text{-}FLOW}, and \kw{DSM\text{-}FLOW+}, and two decomposition algorithms, \kw{Flow\text{-}Decomp} and \kw{Reverse\text{-}Decomp}. We further introduced sandwich bounds to reduce the computation to equivalent smaller hypergraphs and developed incremental algorithms for maintaining density levels under hyperedge insertions and deletions.
Extensive experiments on nine real-world hypergraphs demonstrate the efficiency and scalability of our algorithms. In particular, flow-based search substantially reduces the cost of hyperpath reorientation, while reverse propagation and sandwich acceleration improve decomposition efficiency. Incremental maintenance further reduces update time while producing the same results as full recomputation. Comparisons with established decomposition models and two case studies show that our model captures fine-grained density structures that conventional degree-based methods may overlook.

\section*{Statements and Declarations}
\stitle{Funding.}
This work was partially supported by the New Generation Artificial Intelligence National Science and Technology Major Project (2025ZD0123403), the National Natural Science Foundation of China (NSFC) under Grants 62572055 and 62372342, and the BIT Kunpeng\&Ascend Center of Cultivation. 

\stitle{Competing interests.}
The authors have no relevant financial or non-financial interests to disclose.

\bibliographystyle{spmpsci}
\bibliography{lxy}

@inproceedings{DBLP:conf/cikm/HuWC17,
  author       = {Shuguang Hu and
                  Xiaowei Wu and
                  T.{-}H. Hubert Chan},
  title        = {Maintaining Densest Subsets Efficiently in Evolving Hypergraphs},
  booktitle    = {CIKM},
  pages        = {929--938},
  year         = {2017},
}

@inproceedings{flow,
	author       = {Chrysanthi Kosyfaki and
	Nikos Mamoulis and
	Evaggelia Pitoura and
	Panayiotis Tsaparas},
	title        = {Flow Motifs in Interaction Networks},
	booktitle    = {EDBT},
	pages        = {241--252},
	year         = {2019},
}

@article{bianconi2024nature,
  title={Nature of hypergraph k-core percolation problems},
  author={Bianconi, Ginestra and Dorogovtsev, Sergey N},
  journal={Physical Review E},
  volume={109},
  number={1},
  pages={014307},
  year={2024},
}

@article{DBLP:journals/pvldb/ZhangLZQW24,
  author       = {Yalong Zhang and
                  Ronghua Li and
                  Qi Zhang and
                  Hongchao Qin and
                  Guoren Wang},
  title        = {Efficient Algorithms for Density Decomposition on Large Static and
                  Dynamic Graphs},
  journal      = {Proc. {VLDB} Endow.},
  volume       = {17},
  number       = {11},
  pages        = {2933--2945},
  year         = {2024},
}

@article{DBLP:journals/pacmmod/ZhangLZQQW25,
  author       = {Yalong Zhang and
                  Rong{-}Hua Li and
                  Qi Zhang and
                  Hongchao Qin and
                  Lu Qin and
                  Guoren Wang},
  title        = {Density Decomposition of Bipartite Graphs},
  journal      = {Proc. {ACM} Manag. Data},
  volume       = {3},
  number       = {1},
  pages        = {30:1--30:25},
  year         = {2025},
}

@article{DBLP:journals/corr/cs-DS-0310049,
  author       = {Vladimir Batagelj and
                  Matjaz Zaversnik},
  title        = {An O(m) Algorithm for Cores Decomposition of Networks},
  journal      = {CoRR},
  volume       = {cs.DS/0310049},
  year         = {2003},
}

@article{DBLP:journals/tkde/LiYM14,
  author       = {Rong{-}Hua Li and
                  Jeffrey Xu Yu and
                  Rui Mao},
  title        = {Efficient Core Maintenance in Large Dynamic Graphs},
  journal      = {{IEEE} Trans. Knowl. Data Eng.},
  volume       = {26},
  number       = {10},
  pages        = {2453--2465},
  year         = {2014},
}

@article{DBLP:journals/pvldb/WangC12,
  author       = {Jia Wang and
                  James Cheng},
  title        = {Truss Decomposition in Massive Networks},
  journal      = {VLDB},
  volume       = {5},
  number       = {9},
  pages        = {812--823},
  year         = {2012},
}

@inproceedings{DBLP:conf/sigmod/HuangCQTY14,
  author       = {Xin Huang and
                  Hong Cheng and
                  Lu Qin and
                  Wentao Tian and
                  Jeffrey Xu Yu},
  title        = {Querying k-truss community in large and dynamic graphs},
  booktitle    = {SIGMOD},
  pages        = {1311--1322},
  year         = {2014},
}

@article{DBLP:journals/vldb/ChangW24,
  author       = {Lijun Chang and
                  Zhiyi Wang},
  title        = {A near-optimal approach to edge connectivity-based hierarchical graph
                  decomposition},
  journal      = {{VLDB} J.},
  volume       = {33},
  number       = {1},
  pages        = {49--71},
  year         = {2024},
}

@inproceedings{DBLP:conf/sigmod/ChangYQLLL13,
  author       = {Lijun Chang and
                  Jeffrey Xu Yu and
                  Lu Qin and
                  Xuemin Lin and
                  Chengfei Liu and
                  Weifa Liang},
  title        = {Efficiently computing k-edge connected components via graph decomposition},
  booktitle    = {SIGMOD},
  pages        = {205--216},
  year         = {2013},
}

@article{DBLP:journals/pvldb/ShiDS21,
  author       = {Jessica Shi and
                  Laxman Dhulipala and
                  Julian Shun},
  title        = {Theoretically and Practically Efficient Parallel Nucleus Decomposition},
  journal      = {VLDB},
  volume       = {15},
  number       = {3},
  pages        = {583--596},
  year         = {2021},
}

@inproceedings{DBLP:conf/www/SariyuceSPC15,
  author       = {Ahmet Erdem Sariy{\"{u}}ce and
                  C. Seshadhri and
                  Ali Pinar and
                  {\"{U}}mit V. {\c{C}}ataly{\"{u}}rek},
  title        = {Finding the Hierarchy of Dense Subgraphs using Nucleus Decompositions},
  booktitle    = {WWW},
  pages        = {927--937},
  year         = {2015},
}

@inproceedings{DBLP:conf/sigmod/BonchiKS19,
  author       = {Francesco Bonchi and
                  Arijit Khan and
                  Lorenzo Severini},
  title        = {Distance-generalized Core Decomposition},
  booktitle    = {SIGMOD},
  pages        = {1006--1023},
  year         = {2019},
}

@inproceedings{DBLP:conf/cikm/DaiLQWYZ021,
  author       = {Qiangqiang Dai and
                  Rong{-}Hua Li and
                  Lu Qin and
                  Guoren Wang and
                  Weihua Yang and
                  Zhiwei Zhang and
                  Ye Yuan},
  title        = {Scaling Up Distance-generalized Core Decomposition},
  booktitle    = {CIKM},
  pages        = {312--321},
  year         = {2021},
}

@inproceedings{DBLP:conf/icde/GaoLQCYW22,
  author       = {Sen Gao and
                  Rong{-}Hua Li and
                  Hongchao Qin and
                  Hongzhi Chen and
                  Ye Yuan and
                  Guoren Wang},
  title        = {Colorful h-star Core Decomposition},
  booktitle    = {ICDE},
  pages        = {2588--2601},
  year         = {2022},
}

@article{DBLP:journals/pvldb/GaoQLH23,
  author       = {Sen Gao and
                  Hongchao Qin and
                  Rong{-}Hua Li and
                  Bingsheng He},
  title        = {Parallel Colorful h-star Core Maintenance in Dynamic Graphs},
  journal      = {VLDB},
  volume       = {16},
  number       = {10},
  pages        = {2538--2550},
  year         = {2023},
}

@inproceedings{DBLP:conf/kdd/QinLCZ15,
  author       = {Lu Qin and
                  Rong{-}Hua Li and
                  Lijun Chang and
                  Chengqi Zhang},
  title        = {Locally Densest Subgraph Discovery},
  booktitle    = {SIGKDD},
  pages        = {965--974},
  year         = {2015},
}

@article{DBLP:journals/tkdd/Tatti19,
  author       = {Nikolaj Tatti},
  title        = {Density-Friendly Graph Decomposition},
  journal      = {{ACM} Trans. Knowl. Discov. Data},
  volume       = {13},
  number       = {5},
  pages        = {54:1--54:29},
  year         = {2019},
}

@article{DBLP:journals/pacmmod/ZhangYWLZL25,
  author       = {Wenqian Zhang and
                  Zhengyi Yang and
                  Dong Wen and
                  Wentao Li and
                  Wenjie Zhang and
                  Xuemin Lin},
  title        = {Accelerating Core Decomposition in Billion-Scale Hypergraphs},
  journal      = {Proc. {ACM} Manag. Data},
  volume       = {3},
  number       = {1},
  pages        = {6:1--6:27},
  year         = {2025},
}

@article{DBLP:journals/pvldb/ArafatKRG23,
  author       = {Naheed Anjum Arafat and
                  Arijit Khan and
                  Arpit Kumar Rai and
                  Bishwamittra Ghosh},
  title        = {Neighborhood-based Hypergraph Core Decomposition},
  journal      = {VLDB},
  volume       = {16},
  number       = {9},
  pages        = {2061--2074},
  year         = {2023},
}

@inproceedings{DBLP:conf/ipps/RamadanTP04,
  author       = {Emad Y. Ramadan and
                  Arijit Tarafdar and
                  Alex Pothen},
  title        = {A Hypergraph Model for the Yeast Protein Complex Network},
  booktitle    = {IPDPS},
  year         = {2004},
}

@inproceedings{DBLP:conf/icde/LuoYLZ0023,
  author       = {Qi Luo and
                  Dongxiao Yu and
                  Yu Liu and
                  Yanwei Zheng and
                  Xiuzhen Cheng and
                  Xuemin Lin},
  title        = {Finer-Grained Engagement in Hypergraphs},
  booktitle    = {ICDE},
  pages        = {423--435},
  year         = {2023},
}

@inproceedings{DBLP:conf/cikm/KimKLJ23,
  author       = {Dahee Kim and
                  Junghoon Kim and
                  Sungsu Lim and
                  Hyun Ji Jeong},
  title        = {Exploring Cohesive Subgraphs in Hypergraphs: The (k, g)-core Approach},
  booktitle    = {CIKM},
  pages        = {4013--4017},
  year         = {2023},
}

@article{DBLP:journals/datamine/BuLS23,
  author       = {Fanchen Bu and
                  Geon Lee and
                  Kijung Shin},
  title        = {Hypercore decomposition for non-fragile hyperedges: concepts, algorithms,
                  observations, and applications},
  journal      = {Data Min. Knowl. Discov.},
  volume       = {37},
  number       = {6},
  pages        = {2389--2437},
  year         = {2023},
}

@inproceedings{DBLP:conf/cikm/Luo0000Y024,
  author       = {Qi Luo and
                  Wenjie Zhang and
                  Zhengyi Yang and
                  Dong Wen and
                  Xiaoyang Wang and
                  Dongxiao Yu and
                  Xuemin Lin},
  title        = {Hierarchical Structure Construction on Hypergraphs},
  booktitle    = {CIKM},
  pages        = {1597--1606},
  year         = {2024},
}

@article{jeavons1994structural,
  title={A structural decomposition for hypergraphs},
  author={Jeavons, Peter and Cohen, David and Gyssens, Marc},
  journal={Contemporary Mathematics},
  volume={178},
  pages={161--177},
  year={1994},
}

@article{DBLP:journals/jcss/GottlobLS02,
  author       = {Georg Gottlob and
                  Nicola Leone and
                  Francesco Scarcello},
  title        = {Hypertree Decompositions and Tractable Queries},
  journal      = {J. Comput. Syst. Sci.},
  volume       = {64},
  number       = {3},
  pages        = {579--627},
  year         = {2002},
}

@article{DBLP:journals/jcss/CohenJG08,
  author       = {David A. Cohen and
                  Peter Jeavons and
                  Marc Gyssens},
  title        = {A unified theory of structural tractability for constraint satisfaction
                  problems},
  journal      = {J. Comput. Syst. Sci.},
  volume       = {74},
  number       = {5},
  pages        = {721--743},
  year         = {2008},
}

@article{DBLP:journals/corr/abs-1711-04506,
  author       = {Martin Grohe and
                  D{\'{a}}niel Marx},
  title        = {Constraint Solving via Fractional Edge Covers},
  journal      = {{ACM} Trans. Algorithms},
  volume       = {11},
  number       = {1},
  pages        = {4:1--4:20},
  year         = {2014},
}

@inproceedings{abcore,
author = {Ding, Danhao and Li, Hui and Huang, Zhipeng and Mamoulis, Nikos},
title = {Efficient Fault-Tolerant Group Recommendation Using alpha-beta-core},
year = {2017},
booktitle = {CIKM},
pages = {2047–2050},
}

@INPROCEEDINGS{9458645,
  author={Luo, Qi and Yu, Dongxiao and Cai, Zhipeng and Lin, Xuemin and Cheng, Xiuzhen},
  booktitle={ICDE}, 
  title={Hypercore Maintenance in Dynamic Hypergraphs}, 
  year={2021},
  pages={2051-2056},
}

@article{DBLP:journals/pvldb/QinZLLYW25,
  author       = {Hongchao Qin and
                  Guang Zeng and
                  Ronghua Li and
                  Longlong Lin and
                  Ye Yuan and
                  Guoren Wang},
  title        = {Truss Decomposition in Hypergraphs},
  journal      = {Proc. {VLDB} Endow.},
  volume       = {18},
  number       = {7},
  pages        = {2185--2197},
  year         = {2025},
}

@article{DBLP:journals/corr/abs-2301-04235,
  author       = {Marco Mancastroppa and
                  Iacopo Iacopini and
                  Giovanni Petri and
                  Alain Barrat},
  title        = {Hyper-cores promote localization and efficient seeding in higher-order
                  processes},
  journal      = {CoRR},
  volume       = {abs/2301.04235},
  year         = {2023},
}

@article{DBLP:journals/corr/abs-2204-05646,
  author       = {Martina Contisciani and
                  Federico Battiston and
                  Caterina De Bacco},
  title        = {Principled inference of hyperedges and overlapping communities in
                  hypergraphs},
  journal      = {CoRR},
  volume       = {abs/2204.05646},
  year         = {2022},
}

@inproceedings{DBLP:conf/nips/AltmanBNEAA23,
  author       = {Erik R. Altman and
                  Jovan Blanusa and
                  Luc von Niederh{\"{a}}usern and
                  Beni Egressy and
                  Andreea Anghel and
                  Kubilay Atasu},
  title        = {Realistic Synthetic Financial Transactions for Anti-Money Laundering
                  Models},
  booktitle    = {NeurIPS},
  year         = {2023},
}

@article{DBLP:journals/tpds/MontresorPM13,
  author       = {Alberto Montresor and
                  Francesco De Pellegrini and
                  Daniele Miorandi},
  title        = {Distributed k-Core Decomposition},
  journal      = {{IEEE} Trans. Parallel Distributed Syst.},
  volume       = {24},
  number       = {2},
  pages        = {288--300},
  year         = {2013},
}

@article{cohen2008trusses,
  title={Trusses: Cohesive subgraphs for social network analysis},
  author={Cohen, Jonathan},
  journal={National security agency technical report},
  volume={16},
  number={3.1},
  pages={1--29},
  year={2008},
}

@techreport{goldberg1984finding,
  title={Finding a maximum density subgraph},
  author={Goldberg, Andrew V},
  institution={EECS Department, University of California, Berkeley},
  number={UCB/CSD-84-171},
  year={1984},
}

@inproceedings{charikar2000greedy,
  author       = {Moses Charikar},
  title        = {Greedy approximation algorithms for finding dense components in a
                  graph},
  booktitle    = {APPROX},
  pages        = {84--95},
  year         = {2000},
}

@article{bahmani2012densest,
  author       = {Bahman Bahmani and
                  Ravi Kumar and
                  Sergei Vassilvitskii},
  title        = {Densest Subgraph in Streaming and MapReduce},
  journal      = {Proc. {VLDB} Endow.},
  volume       = {5},
  number       = {5},
  pages        = {454--465},
  year         = {2012},
}

@inproceedings{huang2024densest,
  author       = {Yufan Huang and
                  David F. Gleich and
                  Nate Veldt},
  title        = {Densest Subhypergraph: Negative Supermodular Functions and Strongly
                  Localized Methods},
  booktitle    = {WWW},
  pages        = {881--892},
  year         = {2024},
}

@inproceedings{tsourakakis2015k,
  title={The k-clique densest subgraph problem},
  author={Tsourakakis, Charalampos},
  booktitle={WWW},
  pages={1122--1132},
  year={2015},
}

@article{marx2013tractable,
  author       = {D{\'{a}}niel Marx},
  title        = {Tractable Hypergraph Properties for Constraint Satisfaction and Conjunctive
                  Queries},
  journal      = {J. {ACM}},
  volume       = {60},
  number       = {6},
  pages        = {42:1--42:51},
  year         = {2013},
}

@article{gottlob2021complexity,
  author       = {Georg Gottlob and
                  Matthias Lanzinger and
                  Reinhard Pichler and
                  Igor Razgon},
  title        = {Complexity Analysis of Generalized and Fractional Hypertree Decompositions},
  journal      = {J. {ACM}},
  volume       = {68},
  number       = {5},
  pages        = {38:1--38:50},
  year         = {2021},
}

@misc{xgiDblp2026,
  author={Nicholas Landry},
  title={{coauth-dblp}: {HIF} Dataset, Version 0.2},
  year={2026},
  howpublished={Zenodo},
}

@misc{xgiMath2026,
  author={Nicholas Landry},
  title={{threads-math-sx}: {HIF} Dataset, Version 0.1},
  year={2026},
  howpublished={Zenodo},
}

@article{lee2023k,
  title={(k, q)-core decomposition of hypergraphs},
  author={Lee, Jongshin and Goh, Kwang-Il and Lee, Deok-Sun and Kahng, B},
  journal={Chaos, Solitons \& Fractals},
  volume={173},
  pages={113645},
  year={2023},
}

@article{DBLP:journals/vldb/LuoYZZYCLG26,
  author  = {Qi Luo and Zhengyi Yang and Wenjie Zhang and Alexander Zhou and Dongxiao Yu and Xiuzhen Cheng and Xuemin Lin and Song Guo},
  title   = {Hypergraph Decomposition with Intersection Bounds},
  journal = {{VLDB} J.},
  volume  = {35},
  number  = {4},
  pages   = {34},
  year    = {2026},
}

@article{DBLP:journals/vldb/LuoZYYLW25,
  author  = {Qi Luo and Wenjie Zhang and Zhengyi Yang and Dongxiao Yu and Xuemin Lin and Liping Wang},
  title   = {Efficient Indexing and Searching of Constrained Core in Hypergraphs},
  journal = {{VLDB} J.},
  volume  = {34},
  number  = {3},
  pages   = {34},
  year    = {2025},
}

@article{DBLP:journals/www/LiuLZCY25,
  author  = {Yu Liu and Qi Luo and Yanwei Zheng and Zhipeng Cai and Dongxiao Yu},
  title   = {Efficient Traversal for Core Maintenance in Large-Scale Dynamic Hypergraphs},
  journal = {World Wide Web},
  volume  = {28},
  number  = {4},
  pages   = {45},
  year    = {2025},
}

@inproceedings{DBLP:conf/edbt/LiuLZ00Y25,
  author    = {Yu Liu and Qi Luo and Yanwei Zheng and Wenjie Zhang and Xuemin Lin and Dongxiao Yu},
  title     = {Effective and Efficient Community Search over Large-Scale Hypergraphs},
  booktitle = {{EDBT}},
  pages     = {722--734},
  year      = {2025},
}

@inproceedings{DBLP:conf/edbt/KimKHKJK26,
  author    = {Song Kim and Dahee Kim and Taejoon Han and Junghoon Kim and Hyun Ji Jeong and Jungeun Kim},
  title     = {Efficient Locality-Based Indexing for Cohesive Subgraphs Discovery in Hypergraphs},
  booktitle = {{EDBT}},
  pages     = {170--182},
  year      = {2026},
}

@inproceedings{DBLP:conf/cikm/KimKKK25,
  author    = {Hyewon Kim and Minseok Kim and Dahee Kim and Junghoon Kim},
  title     = {When User Engagement Meets Structural Cohesiveness: A Decay-Driven Approach to Hypergraph Cores},
  booktitle = {{CIKM}},
  pages     = {4879--4883},
  year      = {2025},
}

@article{DBLP:journals/kbs/KimKKKKLL26,
  author  = {Dahee Kim and Hyewon Kim and Song Kim and Minseok Kim and Junghoon Kim and Yeon{-}Chang Lee and Sungsu Lim},
  title   = {Uncovering High-Order Cohesive Structures: Efficient $(k,g)$-Core Computation and Decomposition for Large Hypergraphs},
  journal = {Knowl. Based Syst.},
  volume  = {335},
  pages   = {115195},
  year    = {2026},
}

@inproceedings{DBLP:conf/www/Niu0JL25,
  author    = {Yudong Niu and Yuchen Li and Jiaxin Jiang and Laks V. S. Lakshmanan},
  title     = {{SANS}: Efficient Densest Subgraph Discovery over Relational Graphs without Materialization},
  booktitle = {{WWW}},
  pages     = {4250--4263},
  year      = {2025},
}

@article{DBLP:journals/pacmmod/ZhangLLZQW25,
  author  = {Yalong Zhang and Rong{-}Hua Li and Longlong Lin and Qi Zhang and Lu Qin and Guoren Wang},
  title   = {Integral Densest Subgraph Search on Directed Graphs},
  journal = {Proc. {ACM} Manag. Data},
  volume  = {3},
  number  = {3},
  pages   = {176:1--176:26},
  year    = {2025},
}

@article{DBLP:journals/pvldb/ZhouHF26,
  author  = {Yingli Zhou and Taohua Huang and Yixiang Fang},
  title   = {Efficient Locally $h$-Clique Densest Subgraph Discovery via Divide-and-Conquer},
  journal = {Proc. {VLDB} Endow.},
  volume  = {19},
  number  = {10},
  pages   = {2577--2590},
  year    = {2026},
}
\end{document}